\documentclass[12pt]{article}

\usepackage{amsfonts}
\usepackage{amsmath}
\usepackage{amsthm}
\usepackage{graphicx}
\usepackage{subfig}
\usepackage{caption}
\usepackage{xcolor}

\usepackage{hyperref}
\usepackage[capitalise]{cleveref}

\usepackage{float}
\usepackage{natbib}
\usepackage{cases}
\usepackage{booktabs}

\usepackage{algorithm}
\usepackage[noend]{algpseudocode}

\numberwithin{equation}{section}

\def\theenumi{\arabic{enumi}}

\def\theenumii{\alph{enumii}}
\def\p@enumii{\theenumi.}

\def\theenumiii{\arabic{enumiii}}
\def\p@enumiii{(\theenumi)(\theenumii)}

\def\p@enumiv{\p@enumiii.\theenumiii}
\newtheorem{theorem}{Theorem}

\newtheorem{corollary}{Corollary}

\newtheorem{definition}{Definition}

\newtheorem{lemma}{Lemma}

\newtheorem{proposition}{Proposition}
\newtheorem{remark}{Remark}

\numberwithin{theorem}{section}
\numberwithin{acknowledgement}{section}
\numberwithin{algorithm}{section}
\numberwithin{axiom}{section}
\numberwithin{case}{section}
\numberwithin{claim}{section}
\numberwithin{conclusion}{section}
\numberwithin{condition}{section}
\numberwithin{conjecture}{section}
\numberwithin{corollary}{section}
\numberwithin{criterion}{section}
\numberwithin{definition}{section}
\numberwithin{example}{section}
\numberwithin{exercise}{section}
\numberwithin{lemma}{section}
\numberwithin{notation}{section}
\numberwithin{problem}{section}
\numberwithin{proposition}{section}
\numberwithin{remark}{section}
\numberwithin{solution}{section}
\numberwithin{summary}{section}

\newcommand{\dd}{\mathop{}\!\mathrm{d}} % Derivatives
\def\ee{\mathrm{e}} % Euler's constant

\newcommand{\tciFourier}{\mathcal{F}}

\usepackage{enumitem}

\begin{document}

\title{Arbitrage-free catastrophe reinsurance valuation for compound dynamic contagion claims}
    
\author{Jiwook Jang\footnote{Department of Actuarial Studies \& Business Analytics, Macquarie
Business School, Macquarie University, Sydney NSW 2109, Australia}~~ Patrick J. Laub\footnote{School of Risk and Actuarial Studies, UNSW Business School,
University of New South Wales, Sydney, NSW 2052, Australia, E-mail:
p.laub@unsw.edu.au}~~ Tak Kuen Siu\footnote{Department of Actuarial Studies \& Business Analytics, Macquarie
Business School, Macquarie University, Sydney NSW 2109, Australia}
~~ Hongbiao Zhao\footnote{School of Statistics and Data Science, Shanghai University of Finance and Economics, No. 777 Guoding Road, Shanghai  200433, China}
}
\maketitle

\begin{abstract}
In this paper, we consider catastrophe stop-loss reinsurance valuation for a reinsurance company with dynamic contagion claims. To deal with conventional and emerging catastrophic events, we propose the use of a compound dynamic contagion process for the catastrophic component of the liability. Under the premise that there is an absence of arbitrage opportunity in the market, we obtain arbitrage-free premiums for these contracts. To this end, the Esscher transform is adopted to specify an equivalent martingale probability measure. We show that reinsurers have various ways of levying the security loading on the net premiums to quantify the catastrophic liability in light of the growing challenges posed by emerging risks arising from climate change, cyberattacks, and pandemics. We numerically compare arbitrage-free catastrophe stop-loss reinsurance premiums via the Monte Carlo simulation method. We also compare them with those from generalised compound Hawkes/compound Cox cases. Sensitivity analyses are performed by changing the retention level, the Esscher parameters and the intensity parameters.
\end{abstract}


\section{Introduction}

Huge losses in properties, motor vehicles and from the interruption of businesses occur due to conventional catastrophic events, such as flood, storm, hail, bushfire, earthquake, terrorism, and global war. Emerging cyber and pandemic events bring about catastrophic and far-reaching losses from business interruption.

The estimated cost of claims related to February and March 2022 South-East Queensland and New South Wales floods, which is the third most costly extreme weather event ever recorded in Australia, by the
Insurance Council of Australia (ICA) is AU\$4.8 billion. Insurers' losses from October 2022 Hurricane Ian, one of the most powerful storms ever
recorded in the US, are expected to be between US\$42 billion and US\$57
billion. In July 2023, we have seen extreme weather events, i.e., floods in Asia, wildfires in Canada, and deadly heatwaves in southern Europe and the US, which can impact economies badly with huge financial cost.

Hurricane Helene slammed into Florida's Gulf coast causing deadly flooding, dumping almost unprecedented amounts of rain in September 2024. AccuWeather Inc. estimated that total damage and economic loss could be between US\$145 billion and US\$160 billion. Spain has been affected by record-breaking rainfall and severe flooding in Valencia in late October 2024. The Southern California wildfires of January 2025 illustrate the dynamics of claims. Moody's Risk Management Solutions (RMS) event response estimates that insured losses for Los Angeles firestorm events will likely range between US\$20 billion and US\$30 billion. These events present significant challenges to the financial stability of insurers and reinsurers, highlighting the need for improved models to predict claims from catastrophic events.

In May 2021 a ransomware cyberattack halted the US Colonial Pipeline 
by disrupting its computerised control systems. Nearly half of the US
East Coast's fuel supply depends on this pipeline, which transports gasoline,
diesel, jet fuel and other refined products. The Colonial Pipeline Company acknowledged that it paid a ransom of US\$4.4 million worth of bitcoin, though the US Department of Justice was able to recover 64 bitcoins worth US\$2.3 million. According to Travel, Logistics \& Infrastructure article, airlines suffered \$168
billion in economic losses in 2020, although the COVID-19 pandemic was
entering its endemic stages in some parts of the world \citep{bouwer2022taking}.

Insurers provide coverage to individuals and organisations by underwriting policies that protect against economic and financial losses resulting from catastrophic events. However, insurers cannot cover all losses, as extreme risks pose a challenge to their financial viability. To ensure being solvent in case of exceptional
huge losses, insurers transfer their risk portfolios by purchasing
catastrophe reinsurance contracts. 

The frequency and severity of both conventional and emerging catastrophes are increasing due to factors such as global warming, cyberattacks, and pandemics.
\cite{gurtler2016impact} observed that insurance premiums rose significantly following natural mega-catastrophes to mitigate the risk of financial losses for reinsurers. In March 2024, prior to the Southern California wildfires of January 2025, State Farm, the largest property insurer in the US, announced that it would not renew insurance policies in high-risk California postcodes due to escalating costs, including rising reinsurance premiums and construction expenses.

As a result, alternative point processes are needed to model claim and loss arrivals from both conventional and emerging catastrophic events, ensuring effective pricing of catastrophe insurance and reinsurance contracts. To articulate the ongoing challenge and complexity of new risk dynamics, we use a compound dynamic contagion
process (CDCP) \citep{dassios2011dynamic} for the catastrophic component of the liability of a reinsurer,
(i.e., the aggregate catastrophe loss covered by reinsurance contract), 
denoted by $C_{t}$. This makes our study unique, and to our knowledge,
using it for the catastrophic component of the liability of $C_{t}$ is the first contribution in the context of pricing for
catastrophe reinsurance contract.
For references on this compound point process, we refer to \cite{jang2021review} and relevant citations therein. Recently, \cite{LiuJang2025} applied enhanced dynamic contagion processes to study optimal asset allocation and reinsurance problems.

The CDCP is particularly well suited to modeling catastrophe reinsurance liabilities because its probabilistic structure directly reflects key empirical features of catastrophic loss arrivals. First, the self-exciting component captures claim clustering and loss amplification which are often observed after major catastrophes. Specifically, an initial event, (e.g., a severe flood or wildfire), generates subsequent related claims through infrastructure damage, supply-chain disruption, or aftershocks. Second, the external-excitation component allows for sudden exogenous shocks, such as hurricanes, pandemics, or cyber events, whose arrival times are not triggered by past insurance losses, but by external environmental or systemic factors. Third, the mean-reverting intensity ensures that periods of elevated risk eventually decay, consistent with the transient nature of catastrophe-driven claim activities.

To ensure risk-neutrality, we use an equivalent martingale probability measure $\tilde{\mathbb{P}}$ for the compound dynamic contagion process $C_{t}$, derived using the Esscher transform, to evaluate an arbitrage-free premium.
In the field of actuarial science, risk-neutral pricing
via the Esscher transform \citep{esscher1932probability} can be found, for example, in \cite{GerberShiu1994}, \cite{gerber1996actuarial}, \cite{BDES1996},
\cite{dassios2003pricing}, \cite{jang2004arbitrage}, \cite{STY2004}, and \cite{Siu2005}.

This paper considers the fair valuation of catastrophe stop-loss reinsurance contracts. We
arrange the paper as follows. In \cref{Sec_catastrophe_losses}, we provide a mathematical
definition for the CDCP for $C_{t}$,
which was introduced by \cite{jang2021review}. The infinitesimal generator of the joint process is also provided extending it to the time-inhomogeneous case, together with required moments.
We describe our insurance market and discuss the condition of the absence of arbitrage in \cref{Sec_Pricing}.
\Cref{Sec_Numerical} provides arbitrage-free catastrophe reinsurance premiums numerically obtained using the Monte Carlo simulation method. Comparisons are also made with the generalised Hawkes case and the Cox case, respectively. Sensitivity analyses are performed by changing the retention level, the Esscher parameters and the intensity parameters.
\Cref{Sec_Conclusion} concludes the paper.
Appendix~\ref{EMM Non-zero} defines an equivalent martingale measure under a non-zero interest rate, while Appendix~\ref{simulation} details our Monte Carlo algorithms.

\section{Dynamics for catastrophe losses}\label{Sec_catastrophe_losses}

In this section, the dynamics for catastrophe losses from the insurer/reinsurer's perspective is described. Specifically, a dynamic contagion process (DCP) with a stochastic intensity process is adopted to model the number of catastrophes. Indeed, from the reinsurer's perspective, the catastrophic component of the liability is given by the aggregate catastrophic losses to be covered by a reinsurance contract. The aggregate catastrophic losses are described by a compound
dynamic contagion process (CDCP).

\subsection{Dynamic contagion process}
A definition for the {\it standard} dynamic contagion process (DCP) with {\it time-homogeneous} parameters will be provided by \cref{DCPDefinition}, while a definition for the associated CDCP will be provided by \cref{CDCPDefinition}. We will extend it to a version with {\it time-inhomogeneous} parameters later. 

We start with a complete probability space $(\Omega, {\cal F},\mathbb{P})$, where
$\mathbb{P}$ is a real-world probability measure, which is called a physical probability measure.
Suppose that the probability space $(\Omega, {\cal F},\mathbb{P})$ is sufficiently rich to allow for the definition of all the random processes or quantities that will be introduced subsequently.
A continuous-time model with a finite horizon ${\cal T} := [0, T]$ is considered,
where $T < \infty$. Let $\{ N_t \}_{t \in {\cal T}}$ denote a point process
defined on $(\Omega, {\cal F},\mathbb{P})$. Let $\{ T_{2, j} \}_{j = 1, 2, \ldots}$
be a sequence of random jump times of the point process $\{ N_t \}_{t \in {\cal T}}$.
That is, $\{ T_{2, j} \}_{j = 1, 2, \ldots}$ is a sequence of positive-valued random
variables on $(\Omega, {\cal F},\mathbb{P})$ such that $0 < T_{2, 1} < T_{2, 2} < \cdots$, $\mathbb{P}$-a.s. Then, for all $t \in {\cal T}$,
\begin{equation*} % \label{Nt}
N_t = \sum_{j \ge 1} \mathbb{I}_{\{ T_{2, j} \le t \} }, 
\end{equation*}
where $\mathbb{I}_A$ is the indicator function of an event $A$.

Write $\{ Y_j \}_{j = 1, 2, \ldots}$ for the random positive jump sizes corresponding
to the random jump times $\{ T_{2, j} \}_{j = 1, 2, \ldots}$. Specifically,
for each $j = 1, 2, \ldots$, $Y_j$ is the random positive jump size corresponding
to the random jump time $T_{2, j}$. Assume that $\{ Y_j \}_{j = 1, 2, \ldots}$
is a sequence of independent and identically distributed (i.i.d.) positive
random variables on $(\Omega, {\cal F}, \mathbb{P})$ with the common distribution $G (y)$, where $y > 0$.

Let $\{ M_t \}_{t \in {\cal T}}$ be a Poisson point process on $(\Omega, {\cal F},\mathbb{P})$ such that its intensity process is given by 
$\rho> 0$, for all $t \in {\cal T}$, $\mathbb{P}$-a.s. Write $\{ T_{1, i} \}_{i = 1, 2, \ldots}$
and $\{ X_i \}_{i = 1, 2, \ldots}$ for the random jump times and sizes for the
Poisson point process $\{ M_t \}_{t \in {\cal T}}$, respectively. Again, suppose
that $\{ X_i \}_{i = 1, 2, \ldots}$ is a sequence of i.i.d. positive random
variables on $(\Omega, {\cal F}, \mathbb{P})$ with the common distribution 
$H(x)$, where $x > 0$. Assume that under the probability measure
$\mathbb{P}$, $\{ X_i \}_{i = 1, 2, \ldots}$, $\{ T_{1,i} \}_{i = 1, 2, \ldots}$ and 
$\{ Y_{j} \}_{j = 1, 2, \ldots}$ are independent with each other.

\bigskip

\begin{definition}[Dynamic contagion process] \label{DCPDefinition} Suppose the point process $\{ N_t \}_{t \in {\cal T}}$ has a stochastic intensity process $\{ \lambda_t \}_{t \in {\cal T}}$ under
the probability measure $\mathbb{P}$. Then $\{ N_t \}_{t \in {\cal T}}$ is said to be a dynamic contagion process (DCP) if $\{ \lambda_t \}_{t \in {\cal T}}$ is given by:
\begin{equation} \label{dcp-definition}
\lambda_{t} = a + \left( \lambda_{0}-a\right) \ee^{-\delta t}+\sum_{i \ge
1} X_{i} \ee^{-\delta \left( t - T_{1, i} \right)}\mathbb{I}_{\{ T_{1,i} \le
t \}} + \sum_{j \ge 1} Y_{j} \ee^{-\delta \left(t - T_{2,j} \right)}
\mathbb{I}_{ \{ T_{2, j} \le t \}},
\end{equation}
where
\begin{itemize}
\item the constant mean-reverting (baseline) intensity level $a \ge 0$;
\item the constant initial intensity at time $t = 0$, $\lambda_{0} \geq a$, $\mathbb{P}$-a.s.; 
\item the constant rate of exponential decay $\delta >0$.
\end{itemize}
\end{definition}

In \cref{DCPDefinition}, $\left\{ X_{i}\right\}_{i=1,2,\ldots }$ 
is interpreted as a sequence of positive externally-excited jumps,
while $\left\{ Y_{j}\right\}_{j=1,2,\ldots }$ is interpreted as
a sequence of positive self-excited jumps. If there are no externally excited jumps 
in \cref{dcp-definition} (i.e., $\rho =0$ for all $t \in {\cal T}$) $\{ N_{t} \}_{t \in {\cal T}}$
becomes a generalised Hawkes process. If $a=0$ and there are no
self-excited jumps in \cref{dcp-definition}, $\{ N_{t} \}_{t \in {\cal T}}$ 
becomes the Cox process with shot-noise Poisson intensity. 
It may be noted that these processes are within the general framework of affine processes.
See, for example, \cite{duffie2000transform,duffie2003affine} and \cite{glasserman2010moment}. Let $\mathbb{F}^{\lambda}$
and $\mathbb{F}^N$ denote the $\mathbb{P}$-augmentation of
the natural filtration generated by the stochastic intensity process
$\{ \lambda_t \}_{t \in {\cal T}}$ and the point process
$\{ N_t \}_{t \in {\cal T}}$, respectively. That is, $\mathbb{F}^{\lambda} 
:= \{ {\cal F}^{\lambda}_t \}_{t \in {\cal T}}$ and 
$\mathbb{F}^N := \{ {\cal F}^N_t \}_{t \in {\cal T}}$ so that
for all $t \in {\cal T}$, 
\begin{equation*} %\label{sigmafields}
\begin{aligned} 
{\cal F}^{\lambda}_t &:= \sigma \{ \lambda_u \,|\, u \in [0, t] \} \vee {\cal N}, \\
{\cal F}^N_t &:= \sigma \{ N_u \,|\, u \in [0, t] \} \vee {\cal N},
\end{aligned}
\end{equation*}
where ${\cal N}$ is a collection of the $\mathbb{P}$-null subsets of ${\cal F}$;
${\cal F}_1 \vee {\cal F}_2$ is the minimal $\sigma$-field containing
the two $\sigma$-fields ${\cal F}_1$ and ${\cal F}_2$.
Note that for all $t \in {\cal T}$, ${\cal F}^N_t \subset {\cal F}^{\lambda}_t$.

\subsection{Compound dynamic contagion process}
 This section provides a definition of the compound dynamic contagion process (CDCP) that will be used for the catastrophic component of the liability. We also provide the infinitesimal generator of this process and relevant moments, after which we deal with it separately when the parameters of this process are time-inhomogeneous.

\begin{definition}[Dynamic contagion claims] \label{CDCPDefinition} Suppose $\{ T_{2, j} \}_{j = 1, 2, \ldots}$
is a sequence of random jump times of the DCP $\{ N_t \}_{t \in {\cal T}}$
having the stochastic intensity process $\{ \lambda_t \}_{t \in {\cal T}}$
given by \cref{dcp-definition}. Let $\{ \Xi_{j} \}_{j = 1, 2, \ldots}$ 
be a sequence of i.i.d. positive individual loss amounts with the
common distribution function $J(\xi )$, where $\xi >0$ such that
for each $j = 1, 2, \ldots$, $\Xi_j$ is the individual loss amount
which occurs at random time $T_{2, j}$. Assume
that under the probability measure $\mathbb{P}$, 
$\left\{ X_{i}\right\}_{i=1,2,\ldots }$, $\left\{ T_{1,i}\right\}_{i=1,2,\ldots }$, $\left\{
Y_{j}\right\}_{j=1,2,\ldots }$ and $\left\{ \Xi_{j}\right\}_{j=1,2,\ldots
}$ are independent of each other. Then a compound point process
$\{ C_t \}_{t \in {\cal T}}$ is said to be a compound dynamic contagion process 
(CDCP) if
\begin{equation*}
C_{t}=\sum\limits_{j\geq 1}\Xi_{j}\mathbb{I}_{\{ T_{2,j} \le t \} }. 
\end{equation*}
\end{definition}
Let us denote the integrated intensity process 
\begin{equation*}
\Lambda_t ~:=~ \int^t_0 \lambda_u \dd u.
\end{equation*}

\subsubsection{Time-homogeneous compound dynamic contagion process}

For  the {\it standard} DCP with {\it time-homogeneous} parameters, the infinitesimal generator ${\cal A}$ of the joint process $\{ \left( \lambda_{t},N_{t},C_{t}, M_t, \Lambda_t, t\right) \}_{t \in {\cal T}}$ acting on a function $f(\lambda ,n,c, m, \Lambda,t)\in \mathcal{D}( \mathcal{A} )$ is given by:
\begin{align} \label{eq:infinitesimal_generator}
&\hspace{-0.5cm}\mathcal{A} f( \lambda, n,c, m, \Lambda, t) \notag \\
&= \frac{\partial f}{\partial t}+\delta \left( a-\lambda \right) \frac{\partial f}{\partial
\lambda }  +  \lambda \frac{\partial f}{\partial \Lambda}
 \notag \\
&\quad +\lambda \left[ \int_{0}^{\infty }\int_{0}^{\infty }f(
\lambda +y,n+1,c+\xi , m, \Lambda,t) \dd G(y)\dd J(\xi )-f( \lambda ,n,c, m, \Lambda,t) %
\right]  \notag \\
&\quad +\rho \left[ \int_{0}^{\infty }f( \lambda +x,n,c, m+1, \Lambda,t)
\dd H(x)-f( \lambda ,n,c, m, \Lambda,t) \right],
\end{align}
where $\mathcal{D}( \mathcal{A}) $ is the domain of the generator $\mathcal{A}$ such that $f(\lambda ,n,c,m,\Lambda,t)$ is differentiable
with respect to $\lambda $ and $t$, and
\begin{equation*}
\left\vert \int_{0}^{\infty }\int_{0}^{\infty }f(
\lambda +y,n+1,c+\xi ,m,\Lambda,t) \dd G(y)\dd J(\xi )-f( \lambda ,n,c,m,\Lambda,t)
\right\vert <\infty \text{,}
\end{equation*}
\begin{equation*}
\left\vert
 \int_{0}^{\infty }f( \lambda +x,n,c,m+1,\Lambda,t)
\dd H(x)-f( \lambda ,n,c,m,\Lambda,t)
\right\vert <\infty \text{.}
\end{equation*}
For details on the definition and derivations of the infinitesimal generator in \eqref{eq:infinitesimal_generator}, please refer to
\cite{dassios2011dynamic,dassios2017generalized}, \cite{jang2021review}
and \cite{oksendal2013stochastic}.

Based on \cref{eq:infinitesimal_generator}, we can easily obtain the moments of $\lambda_t$, $N_t$ and $C_{t}$. To do so, we denote the first-order moments of $X$, $Y$ and $\Xi $, respectively, by
\begin{equation*}
\mu_{H}=\int_{0}^{\infty }x \dd H(x),\ \ 
\mu_{G}=\int_{0}^{\infty }y \dd G(y),\ \ 
\mu_{J}=\int_{0}^{\infty }\xi \dd J(\xi) .
\end{equation*}%
and their Laplace transforms by
\begin{equation} \label{Laplacetransforms}
\hat{h}(s) = \int_{0}^{\infty }\ee^{-s
x} \dd H(x), \quad
\hat{g}(s)
= \int_{0}^{\infty }\ee^{-s y} \dd G(y), \quad
\hat{j}(s) = \int_{0}^{\infty }\ee^{-s \xi
}\dd J(\xi ),
\end{equation}%
which are supposed to be finite.
The Laplace transforms will be used in \cref{Sec_Pricing}.
\\

The following proposition for a {\it standard} CDCP with {\it time-homogeneous} parameters is directly adapted from \cite{dassios2011dynamic,dassios2017generalized} and \cite{jang2021review}.

\begin{proposition}[Expectations of $\lambda_t$, $N_t$, and $C_t$ for time-homogeneous CDCP]\label{1D_E_C_CE_thm}
The  expectation of $\lambda_{t}$ conditional on $\lambda_0$ under the real-world probability measure $\mathbb{P}$ is given by
\begin{equation*} %\label{lambda-moment}
\mathbb{E}\left[ \lambda_{t} \mid \lambda_0 \right]
~=~
\left\{
\begin{array}{ll}
\frac{\rho \mu_H + a\delta}{\kappa}
+ \left( \lambda_0 - \frac{\rho \mu_H + a\delta}{\kappa} \right)
\ee^{-\kappa t},
&  \kappa \neq 0,\\
\lambda_0 + \left(\rho \mu_H+ a\delta\right)t,
& \kappa=0,
\end{array} \right.
\end{equation*}
where $\kappa := \delta - \mu_{G}$.
Under the stationary condition $\kappa >0$, the asymptotic first moment of  $\lambda_{t}$ under $\mathbb{P}$ is given by
\begin{equation*}
\mu_1\ :=\ 
\lim_{t\rightarrow \infty}\mathbb{E}\left[ \lambda_{t} \mid \lambda_0 \right]
\ =\ 
\frac{\rho \mu_H + a\delta}{\kappa}.
\end{equation*}
The expectation of $N_{t}$ conditional on $N_0=0$ and $\lambda_0$ under $\mathbb{P}$ is given by
\begin{equation*} %\label{1D_E_C_N_UE_cond_DCPD}
\mathbb{E}\left[ N_{t} \mid  \lambda_0 \right]
~=~
\left\{
\begin{array}{ll}
\mu_1 t + \left( \lambda_0 - \mu_1 \right) \frac{1}{ \kappa } \left( 1-\ee^{- \kappa t} \right),
&  \kappa \neq 0,\\
 \lambda_0 t + \frac{1}{2}\left(\rho \mu_H+ a\delta\right)t^2,
& \kappa= 0.
\end{array} \right.
\end{equation*}
The expectation of $C_{t}$ conditional on $C_0=0$ and $\lambda_0$ under $\mathbb{P}$ is given by
\begin{equation*} %\label{1D_E_C_C_UE_condCDCPD}
\mathbb{E}\left[ C_{t} \mid  \lambda_0 \right]
~=~
\left\{
\begin{array}{ll}
\mu_J [ \mu_1 t + \left( \lambda_0 - \mu_1 \right) \frac{1}{ \kappa } \left( 1-\ee^{- \kappa t} \right) ],
&  \kappa \neq 0,\\
 \mu_J [ \lambda_0 t + \frac{1}{2}\left(\rho \mu_H+ a\delta\right)t^2 ],
& \kappa= 0.
\end{array} \right.
\end{equation*}
\end{proposition}

\subsubsection{Time-inhomogeneous compound dynamic contagion process}

The previously defined DCP $\lambda_t$ from \cref{dcp-definition} is the result of solving the stochastic differential equation (SDE) 
\begin{equation} \label{dcp-sde-definition}
\lambda_{t} = \lambda_0 + \int_{0}^{t} \delta \left( a - \lambda_s \right) \dd s + \sum_{i \ge 1} X_{i} \, \mathbb{I}_{\{ T_{1,i} \le t \}} + \sum_{j \ge 1} Y_{j} \, \mathbb{I}_{ \{ T_{2, j} \le t \}} \,.
\end{equation}
The parameters of this process are all time-homogeneous.
More generally, we can extend to the {\it time-inhomogeneous} situation by allowing the parameters to vary over time (except the decaying rate $\delta$). Specifically, we impose the following assumptions:
\begin{enumerate}
\item[(A1)] The mean-reverting level $a(t)$ and the rate of externally excited jump arrival $\rho(t)$ are deterministic functions of time $t$;
they are bounded on all intervals $[0, t)$ (no explosions); 
and $\lambda_0 \geq a(t)$.
\item[(A2)] The distribution function of the externally excited jump sizes at any time $t$ is $H(x; t)$ for $x > 0$ with $\mu_H(t) = \int_0^\infty x \, \dd H(x; t)$.
\item[(A3)] The distribution function of the self-excited jump sizes at any time $t$ is $G(y; t)$ for $y > 0$ with $\mu_G(t) = \int_0^\infty y \, \dd G(y; t)$.
\item[(A4)] $a(t)$, $\rho(t)$, $G(y; t)$ and $H(x; t)$ are
Riemann-integrable functions of $t$; they are all positive.
\end{enumerate}
After replacing the time-homogeneous constants in \cref{dcp-sde-definition} with time-varying functions, and solving the SDE, we arrive at the $\lambda_t$ for this new time-inhomogeneous DCP
\begin{equation} \label{idcp-definition}
\lambda_{t} = \lambda_0 \ee^{-\delta t} + \delta \int_{0}^{t} a(s) \ee^{-\delta (t-s)} \, \dd s + \sum_{i \ge
1} X_{i} \ee^{-\delta \left( t - T_{1, i} \right)}\mathbb{I}_{\{ T_{1,i} \le
t \}} + \sum_{j \ge 1} Y_{j} \ee^{-\delta \left(t - T_{2,j} \right)}
\mathbb{I}_{ \{ T_{2, j} \le t \}} \,.
\end{equation}

We allow the model parameters to be time-inhomogeneous or time-varying to reflect evolving exposure and claim dynamics in the insurance and reinsurance portfolios. Specifically, in reality, claim intensity and contagion effects may vary over time due to changes in insured exposure, portfolio composition, seasonal effects, and external risk conditions. Introducing time-inhomogeneity at this stage provides a flexible actuarial framework for modeling claim arrivals, while its role under an equivalent change of measures will become apparent later in \cref{Sec_Pricing}.

For this time-inhomogeneous CDCP, we can obtain the moments of $\lambda_t$, $N_t$ and $C_t$ as below, where the relevant parameters are time-dependent. 
They extend the results in \cref{1D_E_C_CE_thm}.

\begin{proposition}[Expectations of $\lambda_t$, $N_t$ and $C_t$ for time-inhomogeneous CDCP] \label{thm:tdlamda_mean}
For the time-inhomogeneous CDCP, the expectation of $\lambda_t$ conditional on $\lambda_0$ under the real-world probability measure $\mathbb{P}$ is given by
\begin{equation} \label{mean-lambda-td}
\mu_{\lambda}(t)
~:=~
\mathbb{E}\left[ \lambda_{t}\mid \lambda_{0}\right]
~=~ 
 \lambda_{0} \ee^{-\int_{0}^{t} \kappa(s) \dd s}
 + \ee^{-\int_{0}^{t} \kappa(s) \dd s} \int\limits_{0}^{t} \ee^{\int_{0}^{s} \kappa(u) \dd u}
 \big\{ \rho(s) \mu_H(s)+a(s) \delta \big\} \dd s,
\end{equation}
where 
\begin{equation*}
\kappa(t) := \delta - \mu_G(t).
\end{equation*}
The expectation of $N_{t}$ conditional on $N_0=0$ and $\lambda_0$ under $\mathbb{P}$ is given by
\begin{equation}\label{mean-N-td}
\mathbb{E}[ N_{t} \mid \lambda_{0} ] 
=
\int\limits_{0}^{t} \mu_{\lambda}(s) \dd s.
\end{equation}
The expectation of $C_{t}$ conditional on $C_0=0$ and $\lambda_0$ under $\mathbb{P}$ is given by
\begin{equation}\label{mean-C-td}
\mathbb{E}[ C_{t} \mid \lambda_{0} ] = \mu_{J} \int\limits_{0}^{t} \mu_{\lambda}(s) \dd s.
\end{equation}
\end{proposition}

\begin{proof}
The infinitesimal generator ${\cal A}$ of the joint process $\{ \left( \lambda_{t},N_{t},C_{t}, M_t, \Lambda_t, t\right) \}_{t \in {\cal T}}$ acting on a function
$
f(\lambda ,n,c,m,\Lambda,t)\in \mathcal{D}\left( \mathcal{A} \right)
$
for the time-inhomogeneous CDCP
is given by:
\begin{align} \label{eq:infinitesimal_generator_td}
&\hspace{-0.5cm}\mathcal{A} f( \lambda, n,c,m, \Lambda, t) \notag \\
&= 
\frac{\partial f}{\partial t}+\delta \big( a(t)-\lambda \big) \frac{\partial f}{\partial
\lambda } + \lambda \frac{\partial f}{\partial \Lambda}
 \notag \\
&\quad +\lambda \left[ \int_{0}^{\infty }\int_{0}^{\infty }f(
\lambda +y,n+1,c+\xi ,m,\Lambda,t) \dd G(y;t)\dd J(\xi )-f( \lambda ,n,c,m,\Lambda,t) %
\right]  \notag \\
&\quad +\rho(t) \left[ \int_{0}^{\infty }f( \lambda +x,n,c,m+1,\Lambda,t)
\dd H(x;t)-f( \lambda ,n,c,m,\Lambda,t) \right],
\end{align}
where $\mathcal{D}\left( \mathcal{A}\right) $ is the domain of the generator 
$\mathcal{A}$ such that $f(\lambda ,n,c,m,\Lambda,t)$ is differentiable
with respect to $\lambda $ and $t$, and
\begin{equation*}
\left\vert \int_{0}^{\infty }\int_{0}^{\infty }f(
\lambda +y,n+1,c+\xi ,m,\Lambda,t) \dd G(y;t)\dd J(\xi )-f( \lambda ,n,c,m,\Lambda,t)
\right\vert <\infty \text{,}
\end{equation*}
\begin{equation*}
\left\vert 
 \int_{0}^{\infty }f( \lambda +x,n,c,m+1,\Lambda,t)
\dd H(x;t)-f( \lambda ,n,c,m,\Lambda,t)
\right\vert <\infty \text{.}
\end{equation*}
Setting $f( \lambda,n , c, m, \Lambda, t) =\lambda $ in
\cref{eq:infinitesimal_generator_td}, we have%
\begin{equation*}
\mathcal{A}\lambda =- \kappa(t) \lambda +\rho(t) \mu_H(t)+a(t) \delta .
\end{equation*}
For each $t \ge 0$, ${\lambda}^{\dagger}_t := \lambda_{t}-\lambda_{0}-\int_{0}^{t}\mathcal{A}\lambda_{s}\dd s$.
Then for any $s < t$, we have:
\begin{eqnarray*}
\mathbb{E} [\lambda^{\dagger}_t - \lambda^{\dagger}_s  \,\mid\, {\cal F}_s ] 
= \mathbb{E} \left [\lambda_t - \lambda_s - \int_{s}^{t}\mathcal{A}\lambda_{u}\dd u 
  \,\mid\,  {\cal F}_s \right ], \quad \mbox{$\mathbb{P}$-a.s.} \end{eqnarray*}
By the Dynkin formula, since $t < \infty$,
\begin{eqnarray*}
\mathbb{E} [ \lambda_t \,|\, {\cal F}_s ] =
\lambda_s + \mathbb{E} \left [ \int_{s}^{t}\mathcal{A}\lambda_{u}\dd u \ | \ {\cal F}_s  \right ], \quad \mbox{$\mathbb{P}$-a.s.}
\end{eqnarray*}
Consequently, by Fubini's theorem,
\begin{eqnarray*}
\mathbb{E} [\lambda^{\dagger}_t - \lambda^{\dagger}_s \,|\, {\cal F}_s ]
= \mathbb{E} \left [\lambda_t - \lambda_s - \int_{s}^{t}\mathcal{A}\lambda_{u}\dd u
  \,\mid\,  {\cal F}_s \right ] = 0,  \quad \mbox{$\mathbb{P}$-a.s.}
\end{eqnarray*}
This implies that 
\begin{eqnarray*}
\mathbb{E} [\lambda^{\dagger}_t \,|\, {\cal F}_s ]
= \lambda^{\dagger}_s, \quad \mbox{$\mathbb{P}$-a.s.}
\end{eqnarray*}
Since $\{ \lambda^{\dagger}_t \}_{t \ge 0}$ is an $(\mathbb{F}, \mathbb{P})$-martingale, we must have:
\begin{equation*}
\mathbb{E} [ \lambda^{\dagger}_t \,|\, {\cal F}_0 ]
= \mathbb{E} [ \lambda^{\dagger}_t \,|\, \lambda_0 ] = \mathbb{E}\Biggl( \lambda_{t} - \lambda_0  -\int_{0}^{t}\mathcal{A}\lambda
_{s}\dd s \,\mid\, \lambda_{0} \Biggr) = 0,
\quad \mbox{$\mathbb{P}$-a.s.}
\end{equation*}%
Consequently,
\begin{eqnarray*}
\mathbb{E} [\lambda_t \,|\, \lambda_0] = \lambda_0
+ \int^{t}_{0} \mathbb{E} \left [  {\cal A} \lambda_u \,|\, \lambda_0 \right ] \dd u, \quad \mbox{$\mathbb{P}$-a.s.}
\end{eqnarray*}
Then,
\begin{equation*}
\mathbb{E}[ \lambda_{t}\mid \lambda_{0} ]
= \lambda_{0}-\int_{0}^{t} \kappa(s) \, 
\mathbb{E}[ \lambda_{s}\mid \lambda_{0} ] \dd s+\int_{0}^{t}\left\{
\rho(s) \mu_H(s)+a(s) \delta \right\} \dd s.
\end{equation*}
Differentiating it with respect to $t$, we then have the ODE for $\mu_{\lambda}(t)
:=\mathbb{E} \left( \lambda_{t}\mid \lambda_{0}\right)$,
\begin{equation*}
\frac{\dd \mu_{\lambda}(t)}{\dd t}
= - \kappa(t) \mu_{\lambda}(t)
+\rho(t) \mu_H(t)+a(t) \delta .
\end{equation*}
and solving this ODE, we obtain \cref{mean-lambda-td} as its solution.
Using \cref{mean-lambda-td}, given $N_0 = C_0 = 0$, we can derive
\begin{equation*} %\label{mean-n-td}
\mathbb{E}[ N_t \mid \lambda_0 ] = \int_{0}^{t} \mathbb{E}[ \lambda_s \mid \lambda_0 ] \dd s \,  
\end{equation*}
and
\begin{equation*} %\label{mean-c-td}
\mathbb{E}[ C_{t}\mid \lambda_{0} ] = \mu_J \int_0^t \mathbb{E}[ \lambda_s \mid \lambda_0 ] \dd s.
\end{equation*}
\end{proof}

The expressions for \cref{mean-N-td,mean-C-td} would be very long with various simple exponential functions. To save space, we just leave their concise expressions, but we will compare them with corresponding time-homogeneous cases in the context of arbitrage-free catastrophe reinsurance premiums in \cref{Sec_Numerical}.
Appendix~\ref{simulation} contains the algorithm we developed to simulate from this new process.

\section{Insurance market and no-arbitrage}\label{Sec_Pricing}

We consider a liquid reinsurance market in which an insurer can cede part of its aggregate risk to a reinsurer via a reinsurance arrangement at any time. Here the no-arbitrage valuation
approach is adopted to value a reinsurance contract. \cite{sondermann1991reinsurance} 
pioneered the no-arbitrage valuation approach to reinsurance contracts. \cite{dassios2003pricing} proposed the use of the Esscher transform to specify 
an equivalent martingale measure for valuing reinsurance contracts by the
no-arbitrage valuation principle. \cite{jang2004arbitrage} introduced the use of the Esscher transform to determine arbitrage-free premiums for extreme losses. It may also be noted that \citet{WuthrichBuhlmannFurrer2010} discussed linear pricing functionals and market consistent valuation for insurance products under the no-arbitrage principle.
The developments here follow those in, for example, \cite{sondermann1991reinsurance}, \cite{dassios2003pricing} and \cite{jang2004arbitrage}.

Suppose that the evolution of the aggregate risk process $\{ C_t \}_{t \in {\cal T}}$ of the insurer over time is modelled by a CDCP as defined in \cref{CDCPDefinition}. To 
simplify the notation and discussion, we consider a complete filtered 
probability space $(\Omega, {\cal F}, \mathbb{F}, \mathbb{P})$, where
$\mathbb{F}$ is a filtration with respect to which the processes
$\{ N_t \}_{t \in {\cal T}}$, $\{ \lambda_t \}_{t \in {\cal T}}$
and $\{ C_t \}_{t \in {\cal T}}$ are adapted. 
As in Definition 3.1 of \cite{dassios2003pricing}, a definition
of a reinsurance strategy starting at a given time $t \in {\cal T}$
is provided.

\begin{definition} \label{reinsurance}
For each $t \in {\cal T}$, a reinsurance strategy starting at time $t$ 
is an $\mathbb{F}$-predictable process $\{ \phi_u \mid u \in [t, T] \}$ on $(\Omega, {\cal F}, \mathbb{P})$ such that for all $u \in [t, T]$, $\phi_u \in [0, 1]$. Here, $\phi_u$ represents the proportion of the insurer’s liability transferred to the reinsurer at time $u$. Note that ``$\phi_u = 0$" corresponds to no reinsurance and that ``$\phi_u = 1$" corresponds to full reinsurance. Write ${\cal H}_t$ for the space of all reinsurance strategies starting at time $t$.
\end{definition}

Let $\{ P_t \}_{t \in {\cal T}}$ denote an $\mathbb{F}$-adapted
process on $(\Omega, {\cal F}, \mathbb{P})$ such that
for all $t \in {\cal T}$, $P_t$ represents the total amount
of insurance premiums received by the insurer up to and including
time $t$. Then assuming that interest rate is constant, the net surplus process $\{ R_t \}_{t \in {\cal T}}$
from the insurance business is given by: 
\begin{equation*} %\label{surplus}
R_{t}=P_{t} - C_{t}, \quad t \in {\cal T}.
\end{equation*}
The assumption of a constant interest rate simplifies the computation of the premium in \cref{ssec:pricing-stoploss}.

For a given time $t \in {\cal T}$, if the insurer selects 
a reinsurance strategy $\{ \phi_u \mid u \in [t, T] \} \in {\cal H}_t$
at time $t$, then the insurer's final gain at time $T$ is given by:
\begin{equation} \label{gain}
G_{t, T} (\phi) = \int_{t}^{T}\phi_{u} \dd R_{u},
\end{equation}
where the stochastic integral in the right-hand side of \cref{gain} is
interpreted as a Stieltjes integral in a pathwise sense. Here
it is supposed that the reinsurer receives the direct insurer's
premiums for its engagement. 

A strategy $\{\phi_u \mid u \in [t, T]\}$ that allows for profit with no possibility of loss is called an arbitrage strategy. Hence, following \cite{sondermann1991reinsurance}, we give the following definition of an arbitrage (reinsurance) strategy.

\begin{definition} [Arbitrage reinsurance strategy] \label{ArbRe} For each
$t \in {\cal T}$, a reinsurance strategy starting at time 
$t$, say $\{ \phi_u \mid u \in [t, T] \} \in {\cal H}_t$, is
said to be an arbitrage reinsurance strategy if the 
respective insurer's final gain at time $T$ satisfies
\begin{equation*} %\label{arbitragegain}
    G_{t, T} (\phi) \geq 0, \ \mathbb{P}\text{-a.s.} \quad \text{and} \quad \mathbb{E} [ G_{t, T} (\phi ) \,|\, {\cal F}_t ] > 0,
\end{equation*}
where $\mathbb{E} [\cdot \,|\, {\cal F}_t ]$ is the conditional expectation
given ${\cal F}_t$ under the probability measure $\mathbb{P}$.
\end{definition}

As noted, for example, in \cite{dassios2003pricing}, 
using the fundamental theorem of asset pricing (\cite{HK1979} and \cite{HP1981}),
the insurance business characterised by the surplus process
$\{ R_t \}_{t \in {\cal T}}$ admits no arbitrage opportunities if there exists an equivalent
martingale measure $\tilde{\mathbb{P}}$ such that the process $\{ R_{t} \}_{t \in {\cal T}}$ is an $(\mathbb{F}, \tilde{\mathbb{P}})$-martingale. 

In fact, the insurance market is incomplete. Consequently,
there is more than one equivalent martingale measure.
As in, for example, \cite{dassios2003pricing} and \cite{jang2004arbitrage}, the Esscher transform will be employed in \cref{Sec_Numerical} to change probability measures. Consequently, 
an equivalent martingale measure can be selected.
The Esscher transform provides a convenient way to
determine an equivalent martingale measure.
The seminal work by \cite{GerberShiu1994} introduced the use of the Esscher transform to option pricing. 
\cite{BDES1996} and \cite{KS2002} generalised the use of the Esscher transform to select an equivalent martingale measure in a general semimartingale market. Since then, the use of the Esscher transform in option valuation has been widely studied. See, for example, \cite{GerberShiu1994}, \cite{gerber1996actuarial}, \cite{dassios2003pricing}, \cite{jang2004arbitrage}, 
\cite{STY2004}, \cite{ECS2005}, \cite{Siu2005}, \cite{GL2008}, \cite{BES2009}, \cite{EPS2009}, amongst many others. In the following, we recall the definition from \cite{dassios2003pricing}, %and \cite{jang2004arbitrage}
where the Esscher transform was used to select an equivalent martingale measure relevant to arbitrage-free reinsurance strategies in an (incomplete) insurance market. The %two 
definition is standard and is presented here for the sake of completeness.

\subsection{Equivalent martingale measures}

In this section, we extend the measure change for the Cox process with shot-noise intensity \citep{dassios2003pricing} to that for a standard dynamic contagion process \citep{dassios2011dynamic}. For further details, we refer to \cite{sondermann1991reinsurance}, where the no-arbitrage approach to pricing reinsurance contracts was first introduced.

Before introducing the formal notion of an equivalent martingale measure, we emphasize that the arbitrage-free pricing method considered in this paper is understood in the actuarial sense, rather than through replication in a liquid financial market. Specifically, we rule out reinsurance strategies that generate non-negative terminal surplus almost surely with strictly positive expected gain, in accordance with \cref{ArbRe}, and adopt the Esscher transform as a
risk-adjusted pricing principle reflecting risk aversion and solvency considerations. In our modeling setup, the market is incomplete and contingent claims such as catastrophe reinsurance
contracts cannot be replicated via dynamic hedging; valuation and hedging therefore “divorce”, as noted in \cite{madan2022nonlinear} (Chapter~13, Page 171 therein). This makes financial arbitrage arguments based on dynamic replication, such as those in \cite{ArtznerEiseleSchmidt2024}, less relevant in our
setting. Our approach is particularly appropriate for catastrophe reinsurance contracts, which are written over finite (typically annual or multi-year) horizons and are not dynamically
hedgeable, and operates within an actuarial martingale valuation paradigm under which the insurer’s surplus process is a martingale under the Esscher-transformed measure (\cref{Emm}),
thereby ensuring actuarial no-arbitrage pricing for non-traded catastrophic risks.

\begin{definition}[Equivalent martingale measure] \label{Emm}
A probability measure $\tilde{\mathbb{P}}$ is said to be an equivalent
martingale probability measure (i.e., $\tilde{\mathbb{P}}$ is equivalent to $\mathbb{P}$ on ${\cal F}_T$) if it satisfies the following three properties:

\begin{enumerate}[label=(\roman*)]
    \item $\tilde{\mathbb{P}}(A)=0$, iff $\mathbb{P}(A)=0$ for any $A \in {\cal F}_T$;
    \item The Radon--Nikodym derivative $\frac{\dd \tilde{\mathbb{P}}}{\dd \mathbb{P}} \in L^{2}(\Omega, {\cal F}_{T}, \mathbb{P})$,
    where $L^{2}(\Omega, {\cal F}_{T},\mathbb{P})$ is the space of square-integrable, ${\cal F}_T$-measurable, random variables on $(\Omega, \mathbb{P})$.
    \item The surplus process $\{ R_t \}_{t \in {\cal T}}$ is an $(\mathbb{F}, \tilde{\mathbb{P}})$-martingale. In particular, 
        \begin{equation*}
        \tilde{\mathbb{E}}\left[ R_{t}\mid \tciFourier_{s}\right]
        = R_{s},\text{ \ } \tilde{\mathbb{P}}-\textit{a.s}.
        \end{equation*}
    for any $0\leq s\leq t\leq T$, where $\tilde{\mathbb{E}} [\cdot \,|\, {\cal F}_s ]$ denotes the conditional expectation given ${\cal F}_s$ with respect to $\tilde{\mathbb{P}}$.
\end{enumerate}

\end{definition}
To simplify the discussion, given that we assume a constant interest rate, we take the interest rate to be zero. The framework extends straightforwardly to a non-zero constant interest rate by working with discounted surplus processes, as detailed in Appendix A. However, the incorporation of a stochastic interest rate, while it is more realistic, would complicate the modeling framework. Consequently, to simplify the discussion, we do not consider the situation of a stochastic interest rate here.

\begin{remark} It may be noted that \cref{Emm} (iii) can be expressed as follows: assuming $R_{0}=0$, (i.e., the surplus at the initial time $0$ is zero), we have:
\[
\tilde{\mathbb{E}}\!\left(R_{t}\mid \mathcal{F}_{0}\right) 
= \tilde{\mathbb{E}}\!\left(P_{t}-C_{t}\mid \mathcal{F}_{0}\right) 
= \tilde{\mathbb{E}}\!\left(P_{t}\mid \mathcal{F}_{0}\right) 
- \tilde{\mathbb{E}}\!\left(C_{t}\mid \mathcal{F}_{0}\right) 
= 0.
\]
Consequently,
\[
\tilde{\mathbb{E}}\!\left(C_{t}\mid \mathcal{F}_{0}\right) 
= \tilde{\mathbb{E}}\!\left(P_{t}\mid \mathcal{F}_{0}\right) 
= \text{an arbitrage-free insurance premium.}
\]
\end{remark}

In the sequel, we shall adopt the Esscher transform to select an equivalent martingale measure. To do this, we start with 
\cref{thm:martingale} giving a (local)-martingale which will be considered when specifying the Radon--Nikodym derivative for changing probability measures via the Esscher transform.

\begin{theorem} \label{thm:martingale}
Suppose that $\rho > 0$ for all $t \in {\cal T}$.
For constants $\theta, \nu, \phi, \psi$, we have an $(\mathbb{F}, \mathbb{P})$-(local)-martingale
\begin{equation*}
\ee^{K(t)} \theta^{N_t} \ee^{ B(t) \lambda_t} ~ \ee^{-\nu C_t} \ee^{\phi \Lambda_t} \psi^{M_t}, \quad t \in {\cal T},
\end{equation*}
if $\{B(t),K(t)\}_{t \in {\cal T}}$ satisfy the non-linear ordinary differential equations (ODE):
\begin{align*}
 B^{\prime}(t)  - \delta B(t)
+ \theta ~ \hat{j}(\nu) ~ \hat{g}(-B(t))    ~ + \phi -1 &=  0,\\
K^{\prime}(t) +
 a \delta  B(t) +   \rho \left[\psi   \hat{h}(-B(t)) -  1 \right]
&=  0,
\end{align*}
where
${\hat h} (\cdot)$, ${\hat g} (\cdot)$ and ${\hat j} (\cdot)$ are the Laplace transforms defined by \cref{Laplacetransforms}.
\end{theorem}

\begin{proof}
Firstly, we consider a function $f( \lambda_{t},N_{t},C_{t}, M_t, \Lambda_t, t)$ 
defined by:
\begin{equation*} %\label{f}
f( \lambda_{t},N_{t},C_{t}, M_t, \Lambda_t, t)
~=~
\ee^{K(t)} \theta^{N_t}   \ee^{   B(t)  \lambda_t} ~ \ee^{-\nu C_t} \ee^{\phi \Lambda_t} \psi^{M_t},
\end{equation*}
for constants $\theta,\psi,\nu$.

Applying It\^o's differentiation rule to $f$ gives:
\begin{eqnarray*} 
f (\lambda_t, N_t, C_t, M_t, \Lambda_t, t) &=&  f (\lambda_0, N_0, C_0, M_0, \Lambda_0, 0)
+ \int^{t}_{0} \delta (a - \lambda_u) \frac{\partial f}{\partial \lambda} \dd u + \int^{t}_{0} \lambda_u \frac{\partial f}{\partial \Lambda} \dd u \nonumber\\ 
&& + \int^{t}_{0} \lambda_{u-} \bigg ( \int^{\infty}_{0} \int^{\infty}_{0} f (\lambda_{u-} + y, N_{u-} + 1, C_{u-} + \xi, M_{u-}, \Lambda_{u-}, u) \nonumber\\
&& \quad \times \dd G(y) \dd J (\xi) - f (\lambda_{u-}, N_{u-}, C_{u-}, M_{u-}, \Lambda_{u-}, u) \bigg ) \dd u \nonumber\\
&& + \int^{t}_{0} \rho_{u-} \bigg ( \int^{\infty}_{0} f (\lambda_{u-} + x, N_{u-}, C_{u-}, M_{u-} + 1, \Lambda_{u-}, u) \dd H (x) \nonumber\\
&& \quad - f(\lambda_{u-}, N_{u-}, C_{u-}, M_{u-}, \Lambda_{u-}, u) \bigg ) \dd u + {\cal M}_t,
\end{eqnarray*}
where $\{ {\cal M}_t \}_{t \ge 0}$ is an $(\mathbb{F}, \mathbb{P})$-(local)-martingale.

Using \cref{eq:infinitesimal_generator}, we have:
\begin{align*}
&\hspace{-0.5cm} f (\lambda_t, N_t, C_t, M_t, \Lambda_t, t) \nonumber\\
&= f (\lambda_0, N_0, C_0, M_0, \Lambda_0, 0)
+ \int^{t}_{0} {\cal A} f (\lambda_{u-}, N_{u-}, C_{u-}, M_{u-}, \Lambda_{u-}, u) \dd u 
+ {\cal M}_t.
\end{align*}
If ${\cal A} f( \lambda , n,c,m, \Lambda, t)=0$ then 
\[
f (\lambda_t, N_t, C_t, M_t, \Lambda_t, t) =  f (\lambda_0, N_0, C_0, M_0, \Lambda_0, 0) + {\cal M}_t,
\]
i.e., $\{ f (\lambda_t, N_t, C_t, M_t, \Lambda_t, t) \}_{t \in {\cal T}}$ is an $(\mathbb{F}, \mathbb{P})$-(local)-martingale.

\sloppy
Plug $f(\lambda_t, N_t, C_t, M_t, \Lambda_t, t)$ into the generator in \cref{eq:infinitesimal_generator} for the \emph{standard} DCP and set $\mathcal{A} f( \lambda , n,c,m, \Lambda, t)=0$, where
\begin{equation*}
f( \lambda, n,c,m, \Lambda, t)
=
\ee^{K(t)} \theta^{n}  \ee^{  B(t)  \lambda} ~ \ee^{-\nu c} \ee^{\phi \Lambda} \psi^{m},
\end{equation*}
i.e.,
\begin{align*}
0&=
K^{\prime}(t) + B^{\prime}(t) \lambda
+\delta \left( a- \lambda \right) B(t)
+ \lambda \phi
\\
&\quad + \lambda \left[ \theta \int_{0}^{\infty }\int_{0}^{\infty } \ee^{ B(t) y}  \ee^{-\nu \xi} \dd G(y)\dd J(\xi )-1
\right]
+
\rho \left[ \psi \int_{0}^{\infty } \ee^{  B(t) x}
\dd H(x)- 1 \right],
\\
0
&=
\lambda \left\{
 B^{\prime}(t)  - \delta B(t) + \phi
+ \theta \int_{0}^{\infty }\int_{0}^{\infty } \ee^{  B(t) y}  \ee^{-\nu \xi} \dd G(y)\dd J(\xi )-1
\right\} \\
&\quad
+ K^{\prime}(t) + \delta   a  B(t) +
\rho \left[ \psi \int_{0}^{\infty } \ee^{ B(t) x}
\dd H(x)- 1 \right] \\
0
&=
\lambda \bigg\{
 B^{\prime}(t)  - \delta B(t) + \phi
+ \theta \hat{j}(\nu) \hat{g}(-B(t)) -1
\bigg\}
+ K^{\prime}(t) + \delta   a  B(t) +
\rho \left[ \psi \hat{h}(-B(t))- 1 \right],
\end{align*}
where
\begin{equation*}
\hat{j}(\nu) := \int_{0}^{\infty }    \ee^{-\nu \xi}  \dd J(\xi ),
\qquad
\hat{g}(B(t)) := \int_{0}^{\infty }  \ee^{ -B(t) y}  \dd G(y),
\qquad
\hat{h}(B(t)) := \int_{0}^{\infty } \ee^{- B(t) x}
\dd H(x).
\end{equation*}
So, we have the ODEs for $B(t),K(t)$,
\begin{align*}
 B^{\prime}(t)  - \delta B(t)
+ \theta ~ \hat{j}(\nu) ~ \hat{g}(-B(t))    ~ + \phi -1 &=  0,\\
K^{\prime}(t) +
 a \delta  B(t) +   \rho \left[\psi   \hat{h}(-B(t)) -  1 \right]
&=  0.
\end{align*}
Hence, the theorem is proved.
\end{proof}

The following theorem provides the existence and uniqueness results for the solutions to the nonlinear ODEs in \cref{thm:martingale}.

\begin{theorem}\label{thm:ODEs}
Assume that the following conditions hold:
\begin{equation*} %\label{ConditionsODEs}
\theta,\psi\geq 1,
\quad
\nu<0 \quad (\hat{j}(\nu)>1),
\qquad
\phi  =  -(\theta ~  \hat{j}(\nu)  -1) <0,
\qquad
\delta  >\theta ~  \hat{j}(\nu)  \mu_G.
\end{equation*}
Then the following statements hold:
\begin{enumerate}
\item There exists a unique solution for $B(t)$ to the nonlinear ODE in \cref{thm:martingale} with the initial condition $B(0)=b> 0$  and the stationary condition $\delta > \mu_G$ such that
\begin{equation*} %\label{ODEBSolution}
B(t)\ =\ 
\mathcal{G}^{-1}(t), \qquad t \geq0,
\end{equation*}
where
\begin{equation*} %\label{eq:def_function_G}
\mathcal{G}(B) := \int^{B}_b  \frac{\dd u }{f_1( u )}, \qquad B\in \left(b, B^+ \right), 
\end{equation*}
\begin{equation*}
f_1 (B) := \delta B -\theta ~ \hat{j}(\nu) ~ \left( \hat{g}(-B ) -1 \right).
\end{equation*}
and $B^+>0$ is the smallest positive solution to $f_1 (B)=0$.

\item Then, there exists a unique solution  for $K(t)$ to the nonlinear ODE in \cref{thm:martingale} with the boundary condition $K(0) = 0$ such that
\begin{equation*} %\label{ODEKSolution}
K (t) = -a \delta \int_0^t  B(s)\dd s + \rho \int_0^t \left[ 1 - \psi \hat{h}(-B(s)) \right] \dd s.
\end{equation*}
\end{enumerate}
\end{theorem}

\begin{proof} 
Assume that
\begin{equation*}
\theta,\psi\geq 1,
\quad
\nu<0 \quad (\hat{j}(\nu)>1),
\qquad
\phi  =  -(\theta ~  \hat{j}(\nu)  -1) <0,
\qquad
\delta  >\theta ~  \hat{j}(\nu)  \mu_G.
\end{equation*}
Then the nonlinear ODE in \cref{thm:martingale} can be written as:
\begin{equation} \label{NonlinearODE}
\frac{\dd B(t)}{\dd t} =
\delta B - \theta ~ \hat{j}(\nu) ~ \left( \hat{g}(-B ) -1 \right),
\end{equation}
with the initial condition $B(0)=b> 0$ and the stationary condition $\delta > \mu_G$.
Define
\begin{equation} \label{f1}
f_1(B) :=
\delta B -\theta ~ \hat{j}(\nu) ~ \left( \hat{g}(-B ) -1 \right).
\end{equation} 
Then
\begin{equation} \label{f_10}
f_1(0)= 0.
\end{equation}
Also,
\begin{equation} \label{d f_1}
\left. \frac{\dd f_1 (B)}{\dd B} \right|_{B=0} =\ \delta - \theta \hat{j} (\nu) \mu_G\ >\ 0.
\end{equation}
Taking the second-order derivative of $f_1 (B)$ with respect to $B$ gives:
\begin{equation*} %\label{d2 f_1}
\frac{\dd^2 f_1 (B) }{\dd B^2} = - \theta \hat{j} (\nu) \int^{\infty}_{0} y^2 \ee^{B y} \dd G(y) < 0.
\end{equation*}
Then $f_1 (B)$ is strictly concave. This, together with \cref{f_10,d f_1}, imply that
$f_1(B) >0$, for $B\in \left(0, B^+\right)$, where 
$B^+>0$ is the smallest positive solution to $f_1 (B)=0$. 

Taking $B = u$ and $t = \tau$, \cref{NonlinearODE,f1} lead to:
\begin{equation} \label{NODE1}
\frac{\dd u}{f_1 (u)} = \dd \tau, \quad u \in (0, B^+).
\end{equation}
Recall that $B (0) = b$ and $B (t) = B$. Then
integrating both sides of \cref{NODE1} gives:
\begin{equation} \label{NODE1Integral}
\int^{B}_{b} \frac{\dd u}{f_1 (u)} = \int^{t}_{0} \dd \tau = t, \quad B \in (0, B^+).
\end{equation}
Define the left-hand side of \cref{NODE1Integral} as:
\begin{equation} \label{Gfunction}
\mathcal{G}(B) := \int^{B}_b  \frac{\dd u}{f_1( u )}, \qquad B\in \left(b,   B^+ \right).
\end{equation}
Note that $\mathcal{G}(B)$ is a continuous and strictly increasing function for $B \in (b, B^+)$. Consequently, its inverse $\mathcal{G}^{-1}$ exists and is uniquely determined. The function $\mathcal{G} : (b, B^+) \to (0, \infty)$ defined in \cref{Gfunction} is an isomorphism. The inverse $\mathcal{G}^{-1} (t)$ then gives the solution $B(t) = B$ to the nonlinear ODE in \cref{thm:martingale} with the initial condition $B (0) = b > 0$ as follows:
\begin{equation*}
B(t)
~~=~~
\mathcal{G}^{-1}(t) ~~ \in ~~ (b, B^+), \qquad t > 0.
\end{equation*}
Once $B(t)$ is uniquely determined as above, $K(t)$ can be obtained from the second nonlinear ODE in \cref{thm:martingale} with the boundary condition $K(0) = 0$ simply as
\begin{equation*}
K (t) =  -a \delta  \int_0^t  B(s)\dd s  +   \rho \int_0^t   \left[ 1 - \psi   \hat{h}(-B(s))   \right] \dd s.
\end{equation*}
\end{proof}

\sloppy
Recall that $\{\ee^{K(t)} \theta^{N_t} \ee^{   B(t)  \lambda_t} ~ \ee^{-\nu C_t} \ee^{\phi \Lambda_t} \psi^{M_t} \}_{t \in {\cal T}}$ is a (local)-martingale. We suppose that some (suitable) square integrability conditions hold so that $\{\ee^{K(t)} \theta^{N_t} \ee^{   B(t)  \lambda_t} ~ \ee^{-\nu C_t} \ee^{\phi \Lambda_t} \psi^{M_t} \}_{t \in {\cal T}}$ is a square-integrable martingale. We now describe how the probability laws of $N_{t}$ and $C_t$ change after changing probability measures by the Esscher transform.  
Firstly, from \cref{thm:martingale}, 
$\{  \ee^{K(t)} \theta^{N_t}   \ee^{   B(t)  \lambda_t} ~ \ee^{-\nu C_t} \ee^{\phi \Lambda_t} \psi^{M_t} \}_{t \in {\cal T}}$ is an $(\mathbb{F}, \mathbb{P})$-martingale. 
Then, for all $t \in {\cal T}$,
\begin{equation*} %\label{MartingaleProperty}
\mathbb{E}_0 [\ee^{K(t)} \theta^{N_t}   \ee^{   B(t)  \lambda_t} ~ \ee^{-\nu C_t} \ee^{\phi \Lambda_t} \psi^{M_t} ] 
= 
\ee^{K(0)} \theta^{N_0}   \ee^{   B(0)  \lambda_0} ~ \ee^{-\nu C_0} \ee^{\phi \Lambda_0} \psi^{M_0}. 
\end{equation*}
Consider the following (normalised) $\mathbb{F}$-adapted process
$\{\zeta_t \}_{t \in {\cal T}}$:
\begin{equation} \label{DensityProcess}
\zeta_t 
:= 
\frac{
\ee^{K(t)} \theta^{N_t}   \ee^{   B(t)  \lambda_t} ~ \ee^{-\nu C_t} \ee^{\phi \Lambda_t} \psi^{M_t}
}{\mathbb{E}_0 [\ee^{K(t)} \theta^{N_t}   \ee^{   B(t)  \lambda_t} ~ \ee^{-\nu C_t} \ee^{\phi \Lambda_t} \psi^{M_t} ]} 
= 
\frac{\ee^{K(t)} \theta^{N_t}   \ee^{   B(t)  \lambda_t} ~ \ee^{-\nu C_t} \ee^{\phi \Lambda_t} \psi^{M_t}
}{\ee^{K(0)} \theta^{N_0}   \ee^{   B(0)  \lambda_0} ~ \ee^{-\nu C_0} \ee^{\phi \Lambda_0} \psi^{M_0}}.
\end{equation}
It is clear that by its definition in \cref{DensityProcess} and \cref{thm:martingale},
$\{ \zeta_t \}_{t \in {\cal T}}$ is an
$(\mathbb{F}, \mathbb{P})$-martingale. For each $t \in {\cal T}$, $\zeta_t > 0$, $\mathbb{P}$-a.s. 
Furthermore, $\mathbb{E}_0 [\zeta_t] = \zeta_0 = 1$. Consequently, a new probability
measure $\tilde{\mathbb{P}}$, which is equivalent to 
the original probability measure $\mathbb{P}$ on
${\cal F}_T$, can be defined by putting:
\begin{equation*} %\label{RNDerivative}
\frac{d \tilde{\mathbb{P}}}{d \mathbb{P}} \bigg |_{{\cal F}_T} := \zeta_T.
\end{equation*}
Given that $\{ \zeta_t \}_{t \in {\cal T}}$ is a square-integrable $(\mathbb{F}, \mathbb{P})$--martingale, then 
\begin{eqnarray*}
\mathbb{E} \left[ \bigg | \frac{d \tilde{\mathbb{P}}}{d \mathbb{P}} \bigg |^2 \right] = \mathbb{E} [ | \zeta_T |^2] < \infty.
\end{eqnarray*}
Consequently, \cref{Emm} (ii) is satisfied.

We refer to related works employing the exponential change of measure: \cite{PalmowskiRolski2002} for Markov processes, \cite{KS2002} via the cumulant process and Esscher transform, and \cite{PalmowskiPojerThonhauser2025} for ruin probabilities with Hawkes arrivals.

\subsection{Infinitesimal generator under the equivalent martingale measure}
Let $\tilde{\mathcal{A}}$ denote the infinitesimal generator 
of the joint process $\{ ( \lambda_{t},N_{t},C_{t}, M_t, \Lambda_t, t) \}_{t \in {\cal T}}$ under the new probability measure $\tilde{\mathbb{P}}$. The key to study how the probability laws of $N_{t}$ and $C_t$ change after changing probability measures by the Esscher transform is to obtain the infinitesimal generator $\tilde{\mathcal{A}}$. This is to be achieved in \cref{thm:new-generator} and \cref{thm:new-lambda} to be presented in the sequel.  To do so we start with a technical lemma.

\bigskip

\begin{lemma} \label{lem:Dynkin}
Assume that $\tilde{f}( \lambda, n,c,m, \Lambda, t) = \tilde{f}( \lambda, t)$ for all $n$, $c$, $m$ and $\Lambda$, and that $\ee^{-B(t)\lambda_t}$ is a martingale. 
Consider $B(t)>0$ for some $t >0$. Then
\begin{equation} \label{eq:lemma34}
\tilde{\mathcal{A}} \tilde{f} (\lambda, 0) 
= \frac{ \mathcal{A} \bigl(\tilde{f}(\lambda, 0)\, \ee^{{-}B(t)\lambda_t}\bigr) }
       { \ee^{-B(t)\lambda_t} }.
\end{equation}
\end{lemma}

\begin{proof}
The generator of the process $(\lambda_t, t)$ acting on a function $\tilde{f}(\lambda, t)$ 
with respect to the equivalent martingale probability measure is
\begin{equation} \label{eq:gen}
\tilde{\mathcal{A}} \tilde{f} (\lambda, 0) 
= \lim_{t \downarrow 0} 
\frac{\tilde{\mathbb{E}}\!\left[\tilde{f}(\lambda_t, t) \mid \lambda_0 = \lambda\right] - \tilde{f}(\lambda, 0)}{t}.
\end{equation}

We will use 
\[
\frac{\ee^{-B(t)\lambda_t}}{\mathbb{E}\!\left(\ee^{-B(t)\lambda_t}\right)}
\]
as the Radon–Nikodym derivative to define the equivalent martingale probability measure, 
where $\mathbb{E}\!\left(\ee^{-B(t)\lambda_t}\right) < \infty$. 
Hence, the expected value of $\tilde{f}(\lambda_t, t)$ given $\lambda$ with respect 
to the equivalent martingale probability measure is
\begin{equation} \label{eq:RN}
\tilde{\mathbb{E}}\!\left[\tilde{f}(\lambda_t, t) \mid \lambda_0 = \lambda\right] 
= \frac{\mathbb{E}\!\left[\tilde{f}(\lambda_t, t)\, \ee^{-B(t)\lambda_t} \mid \lambda_0 = \lambda\right]}
{\mathbb{E}\!\left(\ee^{-B(t)\lambda_t} \mid \lambda_0 = \lambda\right)}.
\end{equation}

Since the denominator in \eqref{eq:RN} is a martingale, it becomes
\begin{equation} \label{eq:mart}
\tilde{\mathbb{E}}\!\left[\tilde{f}(\lambda_t, t) \mid \lambda_0 = \lambda\right] 
= \frac{\tilde{f}(\lambda, 0)\, \ee^{-B(t)\lambda} + 
\int_0^t \mathbb{E}\!\left[ \mathcal{A} \tilde{f}(\lambda_s, s)\, \ee^{-B(t)\lambda_s} \mid \lambda_0 = \lambda \right] \dd s}
{\ee^{-B(t)\lambda}}.
\end{equation}

Substituting \eqref{eq:mart} into \eqref{eq:gen} gives
\begin{equation} \label{eq:final}
\tilde{\mathcal{A}} \tilde{f}(\lambda, 0) 
= \frac{1}{\ee^{-B(t)\lambda}} \lim_{t \downarrow 0} 
\frac{1}{t} \int_0^t \mathbb{E}\!\left[ \mathcal{A} \tilde{f}(\lambda_s, s)\, \ee^{-B(t)\lambda_s} \mid \lambda_0 = \lambda \right] \dd s.
\end{equation}

Therefore, by Dynkin’s formula \citep[see][]{oksendal2013stochastic}, 
\eqref{eq:lemma34} follows immediately.
\end{proof}

\bigskip

\begin{theorem} \label{thm:new-generator}
Suppose that the following conditions hold:
\begin{equation*} % \label{ConditionsODEs2}
\theta,\psi\geq 1,
\quad
\nu<0 \quad (\hat{j}(\nu)>1),
\qquad
\phi  =  -(\theta ~  \hat{j}(\nu)  -1) <0,
\qquad
\delta  >\theta ~  \hat{j}(\nu)  \mu_G.
\end{equation*}
Write, for all
$t \in {\cal T}$, 
\begin{align*} %\label{eq:new-distributions}
\dd \tilde{G}(y;t) &:= \frac{\ee^{B(t) y}}{\hat{g}(-B(t))}  \dd G(y), &
\dd \tilde{H}(x;t) &:= \frac{\ee^{  B(t)  x}}{\hat{h}(-B(t))} \dd H(x), \\
\dd \tilde{J}(\xi) &:= \frac{\ee^{-\nu \xi}}{\hat{j}(\nu)}  \dd J(\xi ), &
\tilde{\rho}(t) &:= \psi \hat{h}(-B(t))~ \rho.
\end{align*}
Let ${\cal D} ({\cal A})$ and ${\cal D} (\tilde{\mathcal{A}})$ denote the domains of the infinitesimal generators ${\cal A}$ and $\tilde{\mathcal{A}}$, respectively. Then, for any function $\tilde{f} \in {\cal D} ({\cal A}) \cap {\cal D} (\tilde{\mathcal{A}})$,
\begin{align}\label{eq:new-generator}
&\hspace{-0.5cm}\tilde{\mathcal{A}} \tilde{f}( \lambda ,  n,c,m, \Lambda, t) \notag \\
&=
\frac{\partial \tilde{f}}{\partial t}
+  \lambda \frac{\partial \tilde{f}}{\partial \Lambda}
+
\delta \left( a-\lambda \right)\frac{\partial \tilde{f}}{\partial\lambda }
 \notag \\
&\quad
+
\theta \hat{j}(\nu) \hat{g}(-B(t))~ \lambda \left[   \int_{0}^{\infty }\int_{0}^{\infty }
\tilde{f}(\lambda +y,n+1,c+\xi ,m,\Lambda,t)    \dd \tilde{G}(y;t) \dd \tilde{J}(\xi )
-  \tilde{f}
\right]  \notag \\
&\quad
+ \tilde{\rho}(t) \left[ \int_{0}^{\infty }\tilde{f}( \lambda +x,n,c,m+1,\Lambda,t)
\dd \tilde{H}(x;t) - \tilde{f}  \right].
\end{align}
\end{theorem}

\begin{proof}
Set
\begin{equation*}
\ee^{K(t)} \theta^{n}  \ee^{  B(t)  \lambda}  ~ \ee^{-\nu c} \ee^{\phi \Lambda} \psi^{m} \tilde{f}( \lambda , n,c,m, \Lambda, t),
\end{equation*}
into the generator in \cref{eq:infinitesimal_generator} and from \cref{lem:Dynkin}, we have
\begin{align*}
&\hspace{-0.1cm}
\tilde{\mathcal{A}} \tilde{f}( \lambda ,  n,c,m, \Lambda, t) \notag \\
&=
\frac{
\mathcal{A} \bigg\{ \ee^{K(t)} \theta^{n}  \ee^{  B(t)  \lambda}  ~ \ee^{-\nu c} \ee^{\phi \Lambda} \psi^{m} \tilde{f}( \lambda , n,c,m, \Lambda, t) \bigg\}
}{\ee^{K(t)} \theta^{n}  \ee^{  B(t)  \lambda}  ~ \ee^{-\nu c} \ee^{\phi \Lambda} \psi^{m}} \\
&=
\frac{\partial \tilde{f}}{\partial t} + \bigg( K^{\prime}(t) + B^{\prime}(t) \lambda \bigg) \tilde{f}
+
\delta \left( a-\lambda \right) \left( \frac{\partial \tilde{f}}{\partial\lambda } + B(t)  \tilde{f} \right)
+
\lambda \left( \frac{\partial \tilde{f}}{\partial \Lambda} + \phi \tilde{f} \right)
 \notag \\
&\quad + \lambda \left[ \theta \hat{j}(\nu) \hat{g}(-B(t)) \int_{0}^{\infty }\int_{0}^{\infty }
\tilde{f}(\lambda +y,n+1,c+\xi ,m,\Lambda,t)      \frac{\ee^{  B(t)  y}}{\hat{g}(-B(t))}  \dd G(y)  \frac{\ee^{-\nu \xi}}{\hat{j}(\nu)}  \dd J(\xi )
-\tilde{f}
\right]  \notag \\
&\quad +\rho \left[ \psi \hat{h}(-B(t)) \int_{0}^{\infty }\tilde{f}( \lambda +x,n,c,m+1,\Lambda,t)
\frac{\ee^{  B(t)  x}}{\hat{h}(-B(t))} \dd H(x) - \tilde{f}  \right].
\end{align*}
Define
\begin{equation*}
\dd \tilde{G}(y;t):= \frac{\ee^{  B(t)  y}}{\hat{g}(-B(t))}  \dd G(y),
\qquad
\dd \tilde{H}(x;t) := \frac{\ee^{  B(t)  x}}{\hat{h}(-B(t))} \dd H(x),
\qquad
\dd \tilde{J}(\xi ) := \frac{\ee^{-\nu \xi}}{\hat{j}(\nu)}  \dd J(\xi ),
\end{equation*}
then,
\begin{align*}
&\hspace{-0.1cm}
\tilde{\mathcal{A}} \tilde{f}( \lambda ,  n,c,m, \Lambda, t) \notag \\
&=
\frac{\partial \tilde{f}}{\partial t} + \bigg( K^{\prime}(t) + B^{\prime}(t) \lambda \bigg) \tilde{f}
+
\delta \left( a-\lambda \right) \left( \frac{\partial \tilde{f}}{\partial\lambda } + B(t)  \tilde{f} \right)
+
\lambda \left( \frac{\partial \tilde{f}}{\partial \Lambda} + \phi \tilde{f} \right)
 \notag \\
&\quad +
\lambda \left[ \theta \hat{j}(\nu) \hat{g}(-B(t)) \int_{0}^{\infty }\int_{0}^{\infty }
\tilde{f}(\lambda +y,n+1,c+\xi ,m,\Lambda,t)    \dd \tilde{G}(y;t) \dd \tilde{J}(\xi )
-\tilde{f}
\right]  \notag \\
&\quad +\rho \left[ \psi \hat{h}(-B(t)) \int_{0}^{\infty }\tilde{f}( \lambda +x,n,c,m+1,\Lambda,t)
\dd \tilde{H}(x;t) - \tilde{f}  \right]
\\
&=
\bigg[    K^{\prime}(t) +    a \delta B(t)
+ \bigg( B^{\prime}(t)  - \delta   B(t)      +\phi \bigg) \lambda
\bigg]  \tilde{f}
+ \frac{\partial \tilde{f}}{\partial t}
+  \lambda \frac{\partial \tilde{f}}{\partial \Lambda}
+
\delta \left( a-\lambda \right)\frac{\partial \tilde{f}}{\partial\lambda }
 \notag \\
&\quad +
\lambda \left[ \theta \hat{j}(\nu) \hat{g}(-B(t)) \int_{0}^{\infty }\int_{0}^{\infty }
\tilde{f}(\lambda +y,n+1,c+\xi ,m,\Lambda,t)    \dd \tilde{G}(y;t) \dd \tilde{J}(\xi )
-\tilde{f}
\right]  \notag \\
&\quad +\rho \left[ \psi \hat{h}(-B(t)) \int_{0}^{\infty }\tilde{f}( \lambda +x,n,c,m+1,\Lambda,t)
\dd \tilde{H}(x;t) - \tilde{f}  \right] \\
&=
\frac{\partial \tilde{f}}{\partial t}
+  \lambda \frac{\partial \tilde{f}}{\partial \Lambda}
+
\delta \left( a-\lambda \right)\frac{\partial \tilde{f}}{\partial\lambda }
 \notag \\
&\quad +
\theta \hat{j}(\nu) \hat{g}(-B(t)) \lambda \left[   \int_{0}^{\infty }\int_{0}^{\infty }
\tilde{f}(\lambda +y,n+1,c+\xi ,m,\Lambda,t)    \dd \tilde{G}(y;t) \dd \tilde{J}(\xi )
-  \tilde{f}
\right]  \notag \\
&\quad +\psi \hat{h}(-B(t)) ~\rho \left[   \int_{0}^{\infty }\tilde{f}( \lambda +x,n,c,m+1,\Lambda,t)
\dd \tilde{H}(x;t) - \tilde{f} \right].
\end{align*}
Define
\begin{equation*}
\tilde{\rho}(t):=\psi \hat{h}(-B(t))~ \rho,
\end{equation*}
and consequently, \cref{eq:new-generator} is obtained.
\end{proof}

\bigskip

\Cref{thm:new-generator} yields the following:

\begin{enumerate}
\item[(i)] The intensity process $\lambda_t$ has changed to
$\tilde{\lambda}_t := \theta \hat{j}(\nu)\hat{g}(-B(t))\,\lambda_t$,
which is time-dependent;

\item[(ii)] The self-exciting jump size measure $\dd G(y) $ has
changed to $\dd \tilde{G}(y;t):= \frac{\ee^{  B(t)  y}}{\hat{g}(-B(t))}  \dd G(y)$, which now depends on time;

\item[(iii)] The rate of external-exciting jump arrival 
$\rho$ has changed
to $\tilde{\rho}(t):=\psi \hat{h}(-B(t))~ \rho$, which now depends on time;

\item[(iv)] The external-exciting jump size measure $\dd H(x)$ has changed to $\dd \tilde{H}(x;t) := \frac{\ee^{  B(t)  x}}{\hat{h}(-B(t))} \dd H(x)$, which now
depends on time;

\item[(v)] The claim/loss size measure $\dd J\left( \xi \right) $ has changed
to $\dd \tilde{J}(\xi ) := \frac{\ee^{-\nu \xi}}{\hat{j}(\nu)}  \dd J(\xi )$.
\end{enumerate}

The Esscher transform provides a transparent and actuarially consistent mechanism for risk-adjusted pricing of catastrophe reinsurance. Under the transformed probability measure based on the Esscher transformation, the arrival intensity $\lambda_t$, jump size measures $\tilde{G}(\cdot; t)$, $\tilde{H}(\cdot; t)$,
$\tilde{J}(\cdot)$, and external excitation factor $\tilde{\rho}(t)$ are adjusted to reflect aversion to extreme losses and catastrophic risk, jointly determining an arbitrage-free premium consistent with actuarial principles. In particular, (i) scaling of the intensity process represents a frequency loading; (ii) tilted self-exciting jump sizes place greater weights on contagion-driven shocks; (iii) time-dependent adjustment of external-exciting arrivals loads for systemic shocks; (iv) transformed external-exciting jump sizes increase the likelihood of large external shocks; and (v) the Esscher tilting on claim sizes acts as a severity-based loading. Together, these adjustments allow reinsurers to incorporate risk aversion into a granular and interpretable manner, distinguishing between frequency, contagion, external shocks, and severity risks, while preserving a consistent arbitrage-free valuation framework.

\bigskip
The result in \cref{thm:new-generator} can be rewritten
by defining 
\begin{equation*}
\tilde{\lambda}_t
:=
\theta \hat{j}(\nu) \hat{g}(-B(t)) \lambda_t,
\qquad
\tilde{\Lambda}_t
:=
\theta \hat{j}(\nu) \hat{g}(-B(t)) \Lambda_t,
\end{equation*} 
for any time $t \in {\cal T}$, which
is presented in \cref{thm:new-lambda}
below.

\begin{corollary} \label{thm:new-lambda} 
The generator in \cref{eq:new-generator} can be alternatively expressed by
\begin{align} \label{eq:new-lambda}
&\hspace{-0.5cm}
\tilde{\mathcal{A}} \tilde{f}( \tilde{\lambda},  n,c,m, \tilde{\Lambda}, t) \notag \\
&=
\frac{\partial \tilde{f}}{\partial t}
+
\tilde{\lambda}   \frac{\partial \tilde{f}}{\partial \tilde{\Lambda} }
+
\delta \bigg( \theta \hat{j}(\nu) \hat{g}(-B(t)) ~ a  - \tilde{\lambda} \bigg)\frac{\partial \tilde{f}}{\partial \tilde{\lambda} }
 \notag \\
&\quad +
\tilde{\lambda} \left[   \int_{0}^{\infty }\int_{0}^{\infty }
\tilde{f}(\tilde{\lambda} +u,n+1,c+\xi ,m,\tilde{\Lambda},t)    \dd \tilde{G}\left( \frac{u}{\theta \hat{j}(\nu) \hat{g}(-B(t))} ; t \right) \dd \tilde{J}(\xi )
-  \tilde{f}
\right]  \notag \\
&\quad + \tilde{\rho}(t) \left[   \int_{0}^{\infty }\tilde{f}( \tilde{\lambda} + v,n,c,m+1,\tilde{\Lambda},t)
\dd \tilde{H}\left( \frac{v}{\theta \hat{j}(\nu) \hat{g}(-B(t))} ; t\right) -   \tilde{f}  \right].
\end{align}
\end{corollary}

\begin{proof}
Define
\begin{equation*}
\tilde{\lambda}
:=
\theta \hat{j}(\nu) \hat{g}(-B(t)) \lambda,
\qquad
\tilde{\Lambda}
:=
\theta \hat{j}(\nu) \hat{g}(-B(t)) \Lambda,
\end{equation*}
and  we have
\begin{align*}
&\hspace{-0.5cm}\tilde{\mathcal{A}} \tilde{f}( \tilde{\lambda},  n,c,m, \tilde{\Lambda}, t) \notag \\
&=
\frac{\partial \tilde{f}}{\partial t}
+
\tilde{\lambda}   \frac{\partial \tilde{f}}{\partial \tilde{\Lambda} }
+
\delta \bigg( \theta \hat{j}(\nu) \hat{g}(-B(t)) ~ a  - \tilde{\lambda} \bigg)\frac{\partial \tilde{f}}{\partial \tilde{\lambda} }
 \notag \\
&\quad +
\tilde{\lambda} \left[   \int_{0}^{\infty }\int_{0}^{\infty }
\tilde{f}(\tilde{\lambda} +\theta \hat{j}(\nu) \hat{g}(-B(t)) y,n+1,c+\xi ,m,\tilde{\Lambda},t)    \dd \tilde{G}(y;t) \dd \tilde{J}(\xi )
-  \tilde{f}
\right]  \notag \\
&\quad + \tilde{\rho}(t) \left[   \int_{0}^{\infty }\tilde{f}( \tilde{\lambda} + \theta \hat{j}(\nu) \hat{g}(-B(t)) x,n,c,m+1,\tilde{\Lambda},t)
\dd \tilde{H}(x;t) -   \tilde{f}  \right].
\end{align*}
Changing variables $u=\theta \hat{j}(\nu) \hat{g}(-B(t)) y, v=\theta \hat{j}(\nu) \hat{g}(-B(t)) x$, we have \cref{eq:new-lambda}.
\end{proof}

\bigskip

\begin{corollary}\label{thm:new-lambda_2} 
More explicitly, for
\begin{equation*}
\dd \tilde{G}(y;t) = \tilde{g}(y;t) \dd y,
\qquad
\dd \tilde{H}(x;t) = \tilde{h}(x;t) \dd x,
\end{equation*}
we have the transforms
\begin{align*}
a \ &\rightarrow \ \theta \hat{j}(\nu) \hat{g}(-B(t)) ~ a; \\
\rho \ &\rightarrow \  \psi \hat{h}(-B(t))~ \rho ;\\
h(v) \ &\rightarrow \ \frac{\tilde{h}\left( \frac{v}{\theta \hat{j}(\nu) \hat{g}(-B(t))}; t \right)}{\theta \hat{j}(\nu) \hat{g}(-B(t))} ; \\
g(u) \ &\rightarrow \ \frac{ \tilde{g}\left( \frac{u}{\theta \hat{j}(\nu) \hat{g}(-B(t))}; t\right)  }{\theta \hat{j}(\nu) \hat{g}(-B(t))} ; \\
j(\xi) \ &\rightarrow \ \tilde{j}(\xi) ,
\end{align*}
where
\begin{equation*}
\tilde{g}(y;t) := \frac{\ee^{B(t) y}}{\hat{g}(-B(t))} g (y),
\qquad
\tilde{h}(x;t) := \frac{\ee^{B(t) x}}{\hat{h}(-B(t))} h(x),
\qquad
\tilde{j}(\xi) := \frac{\ee^{-\nu \xi}}{\hat{j}(\nu)} j (\xi ).
\end{equation*}
\end{corollary}

\begin{proof}
Based on \cref{thm:new-lambda}, in particular for
\begin{equation*}
\dd \tilde{G}(y;t) = \tilde{g}(y;t) \dd y,
\qquad
\dd \tilde{H}(x;t) = \tilde{h}(x;t) \dd x,
\end{equation*}
we have
\begin{align*}
&\hspace{-0.5cm}\tilde{\mathcal{A}} \tilde{f}( \tilde{\lambda},  n,c,m, \tilde{\Lambda}, t) \notag \\
&=
\frac{\partial \tilde{f}}{\partial t}
+
\tilde{\lambda}   \frac{\partial \tilde{f}}{\partial \tilde{\Lambda} }
+
\delta \bigg( \theta \hat{j}(\nu) \hat{g}(-B(t)) ~ a  - \tilde{\lambda} \bigg)\frac{\partial \tilde{f}}{\partial \tilde{\lambda} }
 \notag \\
&\quad +
\tilde{\lambda} \left[ \int_{0}^{\infty }\int_{0}^{\infty }
\tilde{f}(\tilde{\lambda} +u,n+1,c+\xi ,m,\tilde{\Lambda},t)     \dd \tilde{J}(\xi ) ~ \frac{ \tilde{g}( \frac{u}{\theta \hat{j}(\nu) \hat{g}(-B(t))} ;t)  }{\theta \hat{j}(\nu) \hat{g}(-B(t))} \dd u
-  \tilde{f}
\right]  \notag \\
&\quad + \tilde{\rho}(t) \left[   \int_{0}^{\infty }\tilde{f}( \tilde{\lambda} + v,n,c,m+1,\tilde{\Lambda},t)
    \frac{\tilde{h}( \frac{v}{\theta \hat{j}(\nu) \hat{g}(-B(t))} ;t )}{\theta \hat{j}(\nu) \hat{g}(-B(t))}   \dd  v  -   \tilde{f}  \right].
\end{align*}
Then, we have the transforms.
\end{proof}

\bigskip

Note that the parameters of a CDCP under the original measure $\mathbb{P}$ are all constants, whereas the parameters of a CDCP under the new measure $\tilde{\mathbb{P}}$ become time-varying (except for the intensity-decay rate $\delta$), i.e., a {\it time-inhomogeneous} CDCP.
\\

In reality, a ``hard'' market in a (re)insurance industry occurs when the supply of coverage available is lower than its demand. This leads to a higher (re)insurance premium, reduced (or stopped) coverage, and less competition among (re)insurance providers. The Southern California wildfires of January 2025 further exacerbated this, as rising losses from catastrophic events, coupled with the increasing costs of reinsurance and construction, have prompted insurers like State Farm to limit or withdraw coverage in high-risk areas, intensifying the strain on the market.

There are several reasons why a hardening reinsurance market prevails. Firstly, there have been increasing frequency and severity of losses attributed to catastrophic events due to El Ni\~{n}o and La Ni\~{n}a. Secondly, a high inflation regime has been in force, which may be partly due to turbulence in macroeconomic conditions. Thirdly, there has been a growing number of emerging events such as cyber and pandemic events.

In a hardening reinsurance market, it may be reasonable for a reinsurer to calculate the premium of a catastrophe reinsurance contract in a more conservative and prudent way. To illustrate this situation in our current modelling set up, certain assumptions are imposed on $\theta,\psi, \nu, B(t)$ such that there would be more external-exciting jump arrivals in a given period of time, a higher value of self-exciting intensity, higher values of their jump sizes, as well as a larger size of claim or loss. To take on these additional risks in the hardening market, a reinsurer requires a higher compensation. Furthermore, it may be reasonable to postulate that the reinsurer aims to maximise their shareholders' wealth via earning profits. To this end, the reinsurer may charge a higher premium than the ``break-even'' premium. The latter is given by the expected claims evaluated under the original (real-world) probability measure. In the hardening market, $\theta,\psi, \nu, B(t)$ may be thought of as security loading factors to be evaluated by which positive gross premium should be finally charged. In fact, based on \cref{thm:new-lambda_2}, if the gross premium is always positive, we shall have the conditions
\begin{equation*}
\theta,\psi\geq 1,
\qquad
\nu<0,
\qquad
\mbox{$B(t)>0$ for some $t >0$}.
\end{equation*} 

From an economic perspective, the Esscher parameters $(\theta,\psi,\nu)$ provide a parsimonious
representation of reinsurer risk preferences under hard market conditions. In particular, the
transformed intensity process
\[
\tilde{\lambda}_t := \theta \hat{j}(\nu)\hat{g}(-B(t))\,\lambda_t
\]
introduces a time-dependent loading on claim arrivals, reflecting heightened sensitivity to both
claim frequency and loss severity during periods of elevated risk. The parameter $\theta$ acts as a frequency loading on the baseline arrival intensity, while the factor $\hat{j}(\nu)\hat{g}(-B(t))$ captures the interaction between severity aversion and contagion-driven
dynamics. Similarly, the parameter $\psi$ scales the arrival rate of externally driven shocks, representing increased concern about systemic events. The Esscher parameter $\nu<0$ induces a severity-based loading by tilting the claim size distribution toward larger losses, thereby increasing the contribution of catastrophic events to the premium.

In a hard reinsurance market, these choices of the parameters are consistent with observed outcomes such as higher premiums and tighter underwriting, resulting in a conservative valuation. On the other hand, reinsurers may also limit coverage in high-risk regions by transferring part of the risk whenever alternative capacity is available. In this sense, the Esscher parameters $(\theta,\psi,\nu)$ may serve as reduced-form measures incorporating market risk aversion and capital scarcity into catastrophe reinsurance pricing.

\subsection{Special Case: Exponential/Gamma Distributions for Jump Sizes}

Before closing this section, we show the explicit transformation for the
computation of arbitrage-free catastrophe stop-loss reinsurance premiums in \cref{Sec_Numerical}.

\begin{corollary}[Special Case: Exponential/Gamma Distributions for Jump Sizes]\label{corollary_case_exp_gamma}
In particular, we assume that external-exciting and
self-exciting jump sizes in the intensity process follow exponential distributions $H \sim \mathsf{Exp}(\alpha), G \sim \mathsf{Exp}(\beta)$, respectively, i.e., {%
\begin{equation*}
h(x) =\alpha \ee^{-\alpha x},
\qquad
g(y) =\beta \ee^{-\beta
y},\ 
\qquad\alpha,\beta >0,
\end{equation*}%
and }claim/loss sizes follow  the gamma distribution $J \sim \mathsf{Gamma}(\gamma,\eta)$ with the rate parameter $\gamma$ and shape parameter $\eta$, i.e.,%
\begin{equation*}
j\left( \xi \right) 
=
\frac{\gamma ^{\eta }\xi ^{\eta -1}\ee^{-\gamma \xi}}{\left( \eta -1\right) !},
\qquad 
\gamma >0, \eta \geq 1.
\end{equation*}%
Then, we have the transforms 
\begin{align*}
a  \ &\rightarrow \  \theta\Bigl(1 + \frac{\nu}{\gamma}\Bigr)^{-\eta} \frac{\beta}{\beta - B(t)} a; \\
\rho  \ &\rightarrow \   \psi \frac{\alpha}{\alpha - B(t) } ~ \rho;\\
J \sim \mathsf{Gamma}(\gamma,\eta)\ &\rightarrow \  \mathsf{Gamma}(\gamma+\nu, \eta); \\
H \sim \mathsf{Exp}(\alpha) \ &\rightarrow \ 
\mathsf{Exp}(\lambda_h); \\
G \sim \mathsf{Exp}(\beta)  \ &\rightarrow \ 
\mathsf{Exp}(\lambda_g);
\end{align*}
where 
\begin{equation*}
\lambda_h := \frac{\alpha - B(t)}{\theta \left(1 + \frac{\nu}{\gamma}\right)^{-\eta} \frac{\beta}{\beta-B(t)}},
\qquad 
\lambda_g := \frac{\beta - B(t)}{\theta \left(1 + \frac{\nu}{\gamma}\right)^{-\eta} \frac{\beta}{\beta-B(t)}},
\qquad \nu\in (-\gamma,0).
\end{equation*}
\end{corollary}

\begin{proof}
In particular for exponential distributions, $H \sim \mathsf{Exp}(\alpha), G \sim \mathsf{Exp}(\beta)$, i.e.,
\begin{equation*}
h(x) =\alpha \ee^{-\alpha x},
\qquad
g(y) =\beta \ee^{-\beta y},
\qquad \alpha,\beta >0,
\end{equation*}
we have
\begin{equation*}
\hat{g}(-B(t)) = \frac{\beta}{\beta -B(t)},
\qquad
\hat{h}(-B(t)) = \frac{\alpha}{\alpha -B(t)},
\end{equation*}
\begin{equation*}
\tilde{g}(y)
= (\beta -B(t))    \ee^{-(\beta -B(t)) y},
\qquad
\tilde{h}(x)
=
(\alpha -B(t))    \ee^{-(\alpha -B(t)) x}.
\end{equation*}
Similarly, for gamma distribution $J \sim \mathsf{Gamma}(\gamma,\eta)$, we have the transform
to $\mathsf{Gamma}(\gamma+\nu,\eta)$, since
\begin{align*}
j\left( \xi \right)
&=
\frac{\gamma ^{\eta }\xi ^{\eta -1}\ee^{-\gamma \xi}}{\left( \eta -1\right) !},
\qquad \gamma >0, \eta \geq 1,
\qquad
\hat{j}(\nu) = \left(\frac{\gamma}{\gamma + \nu} \right)^{\eta},
\\
\tilde{j}(\xi )
&=
\frac{\ee^{-\nu \xi}}{\hat{j}(\nu)}    j (\xi )
~~=~~
\frac{\ee^{-\nu \xi}}{\hat{j}(\nu)}    \frac{\gamma ^{\eta }\xi ^{\eta -1}\ee^{-\gamma \xi}}{\left( \eta -1\right) !}
~~=~~
        \frac{ (\gamma + \nu)^{\eta} \xi ^{\eta -1}   \ee^{-(\gamma+\nu) \xi} }{\left( \eta -1\right) !}.
\end{align*}
Finally, we obtain the results by applying the transforms in \cref{thm:new-lambda_2}.
\end{proof}

\bigskip

In fact, the original distributional structure of the underlying risk process via this change of measure is preserved.
We may also need the stationary condition for a CDCP under the new measure, i.e.,
\begin{equation*}
\delta
~~>~~
\frac{1}{\lambda_g}
~~=~~
\theta \left( \frac{\gamma}{\gamma+\nu}\right)^{\eta} \frac{\beta}{(\beta - B(t))^2},
\end{equation*}
where $\delta$ and $\frac{1}{\lambda_g}$ are the decay rate and the mean of exponentially distributed self-excited jump sizes  under the new measure for a CDCP, respectively, according to \cref{corollary_case_exp_gamma}.
So, for implementations, we have to set the decay rate $\delta$ relatively large together with the condition $\delta  >\theta ~  \hat{j}(\nu)  \mu_G =\theta ~  \hat{j}(\nu) \frac{1}{\beta}$ required in \cref{thm:ODEs} for the existence and uniqueness of the solutions of the ODEs to hold.
This constraint on $\delta$ is required in our model. In other words, the model may not be applicable if this constraint on $\delta$ does not hold. This may represent a limitation of the model. \\

Note that, for the ODE of $B(t)$, we have
\begin{equation*}
f_1 (B) = \delta B \left( \frac{B  - \left( \beta - \frac{1}{\delta} \theta ~ \hat{j}(\nu) \right) }{B-\beta} \right),
\end{equation*}
where $B^+:=\beta - \frac{1}{\delta} \theta ~ \hat{j}(\nu)  >0$ or $\delta \beta   > \theta ~ \hat{j}(\nu)$.
Define
\begin{equation} \label{eq:def_function_G_exp}
\mathcal{G}(B) := \int^{B}_b  \frac{\dd u}{f_1(u)}, \qquad B\in \left(b, B^+  \right).
\end{equation}
Note that $\mathcal{G}(B)$ is a strictly increasing function and
\begin{equation*}
\lim_{B \downarrow b}\mathcal{G}(B) = 0,
\qquad
\lim_{B \uparrow B^+ }\mathcal{G}(B) = \infty.
\end{equation*}
We have its inverse as the solution to $B(t)$ as
\begin{equation*}
B(t) \ = \ \mathcal{G}^{-1}(t), \qquad t \geq0.
\end{equation*}
Finally, we have the following three types of parameter settings:
\begin{enumerate}
  \item If $\alpha \geq B^+$, then $B(t) \in (0, \alpha)$ for any time $t \geq0$;

  \item If $\alpha \in (b,B^+)$, then $B(t) \in (0, \alpha)$ for $t \in [0,  t^* )$ where $t^*:=\mathcal{G}(\alpha) >0$ from \cref{eq:def_function_G_exp};

  \item If $\alpha \in (0,b]$, then $B(t) \in (0, \alpha)$ there is no solution for $B(t) \in (0, \alpha),t\geq 0$;
\end{enumerate}
where $b>0, B^+=\beta - \frac{1}{\delta} \theta ~ \hat{j}(\nu) > 0$.
Note that the function $\mathcal{G}(\cdot)$ in \cref{eq:def_function_G_exp} shall be derived fully analytically. However, its inverse $\mathcal{G}^{-1}(\cdot)$ has to be solved numerically.
\\

The first type of parameter setting (i.e., $\alpha \geq B^+$) is the most ideal.
For numerical examples, it is suggested that, $b$ is chosen to be small and close to zero, and $\alpha$ is relatively large.
\\

In particular if $\beta \rightarrow \infty$, then, all self-exciting jump sizes approach to zero, which degenerates to the pure shot-noise intensity model of \cite{dassios2003pricing}.

\begin{corollary}[Expectations for exponential/gamma distributions]
When $H \sim \mathsf{Exp}(\alpha)$, $G \sim \mathsf{Exp}(\beta)$, and $J \sim \mathsf{Gamma}(\gamma,\eta)$, the expectations are as follows:
\begin{equation*}
\begin{aligned}
\tilde{\mathbb{E}}[ \lambda_t \mid \lambda_0]\ &= \ \lambda_0 \ee^{-\int_0^t I(s) \dd s} 
+ \ee^{-\int_0^t I(s) \dd s} \theta \left( \frac{\gamma }{\gamma +\nu }\right) ^{\eta }  
    \\
&\quad \times \int_0^t \ee^{\int_0^s I(u) \dd u}  \left( \frac{\beta}{\beta -B(s)}\right) 
     \left\{ \psi \alpha \rho \left( \frac{1}{\alpha -B(s)}\right)^{2} + a \delta \right\} \dd s.
\end{aligned}
\end{equation*}
\begin{equation*}
\tilde{\mathbb{E}}[ N_t \mid \lambda_0 ] = \int_{0}^{t} \tilde{\mathbb{E}}[ \lambda_{s} \mid \lambda_{0} ] \dd s
\end{equation*}
\begin{equation*}
\tilde{\mathbb{E}}[ C_t \mid \lambda_0 ] = \mu_{\widetilde{J}} \int_0^t \tilde{\mathbb{E}}[ \lambda_s \mid \lambda_0 ] \dd s
\end{equation*}
where $
I(s) = \delta - \left[\theta \left( \frac{\gamma }{\gamma +\nu }\right)^{\eta} \left( \frac{\beta }{\beta -B(s) }\right) \right]/[\beta -B(s)]
$
and $\mu_{\widetilde{J}} = \eta/(\gamma + \nu)$.
\end{corollary}
\begin{proof}
These follow from substituting the transformed values into \cref{thm:tdlamda_mean}.
\end{proof}

For illustrative purposes, we have used a Gamma distribution for claim sizes, which yields a closed-form expression under the Esscher transform, $\tilde{J} \sim \mathsf{Gamma}(\gamma+\nu,\eta)$.
The Gamma family is a standard choice in actuarial modeling of claim severities \citep[Page~7]{buhlmann1970mathematical} and offers analytical tractability that facilitates clear interpretation of the pricing mechanism. While heavier-tailed alternatives, such as the Pareto distribution,
\[
j(\xi) = \frac{\Gamma(\omega + k)\, \zeta^\omega \, \xi^{k-1}}
{\Gamma(\omega)\, \Gamma(k)\, (\zeta + \xi)^{\omega + k}},
\qquad \omega>0,\ \zeta>0,\ k>0,
\]
are commonly employed to model large insurance losses
\citep[Chapters~3--4]{albrecher2017reinsurance},
their form is not preserved under the Esscher transform and instead acquires an exponential tilt. Nevertheless, the corresponding Laplace transform $\hat{j}(\nu)$ exists under suitable
parameter conditions and can be evaluated numerically. Heavier-tailed claim size distributions typically lead to higher premium levels, reflecting increased tail risk. In this sense, the
Gamma specification, including the exponential distribution as a special case, serves as a benchmark for assessing the impact of tail heaviness on catastrophe reinsurance pricing. The results presented here therefore illustrate the qualitative behavior of the model, while extensions to Pareto or other heavy-tailed distributions may be explored in future work via numerical sensitivity analysis.

\section{Numerical results}\label{Sec_Numerical}

In this section, we provide arbitrage-free catastrophe stop-loss reinsurance premiums via the Monte Carlo simulation method.
We also examine the sensitivity analyses for the retention level, the Esscher parameters and the intensity parameters.

\subsection{Pricing of a stop-loss reinsurance contract} \label{ssec:pricing-stoploss}

The total loss excess over $L$, which is a retention limit, up to time $t$ is
$(C_{t}-L)^+$. Then the stop-loss reinsurance gross premium at time 0 (i.e. under $\tilde{\mathbb{P}}$) is 
\begin{equation} \label{stoploss}
        \tilde{\mathbb{E}}[(C_t - L)^+]. %:= \int_{0}^\infty (c - A) \mathbb{P}^{\ast}(\dd c) \,.
\end{equation}
Its Monte Carlo estimate will be denoted by 
$\widehat{\tilde{\mathbb{E}}}[(C_t - L)^+]$, and the corresponding estimate under $\mathbb{P}$ will be denoted by $\widehat{\mathbb{E}}[(C_t - L)^+]$.
For this section, we consider
\begin{equation*}
\alpha = 2, \beta = 1, \eta = 3, \gamma = 0.4 \,,
\end{equation*}
\begin{equation*}
\delta = 3, \rho = 4, a = 1, \lambda_0 = 1, t = 1 \,,
\end{equation*}
\begin{equation*}
\nu = -0.05, \theta = 1.25, \psi = 1.25, \phi = - (\theta \hat{j}(\nu) - 1) \,.
\end{equation*}

The parameter values adopted in the numerical examples are chosen for illustrative purposes. They are used to demonstrate the qualitative behaviour of the model, as well as the impacts of contagion and risk loading on reinsurance premiums, rather than to provide an empirical calibration to a specific portfolio. Similar parameter ranges are commonly adopted in the dynamic contagion and self-exciting point process literature to study clustering effects and catastrophic risk (\cite{dassios2011dynamic}; Jang and Oh (\citeyear{jang2021review}, \citeyear{JangOh2025})).

In practice, model parameters may be estimated from historical catastrophe or large-loss data
using standard inference techniques for point processes with latent or stochastic intensities, including likelihood-based methods, moment-based estimators, and filtering approaches. General statistical frameworks for point process inference are discussed in
\cite{Karr1991} and \cite{DaleyVereJones2008}, while modern likelihood-based methodologies for self-exciting and Hawkes-type processes are presented in \cite{Laub2022}. Related filtering approaches in insurance risk models are studied in \cite{DassiosJang2005} for shot-noise Cox processes, and similar ideas may be adapted to the dynamic contagion setting considered here.

While maximum likelihood estimators for point process parameters and some of their statistical properties have been studied, the estimation of shot-noise Cox processes, and in particular the analysis of their asymptotic properties, has received comparatively less attention. These challenges are further compounded when time-varying parameters are present, suggesting promising opportunities for future research, particularly in the context of dynamic contagion models for insurance risk.

For each combination of parameters, $10^5$ crude Monte Carlo (CMC) samples of $C_t$ under $\tilde{\mathbb{P}}$ have been used to estimate \cref{stoploss} in \cref{tbl:dcp}. The corresponding values under $\mathbb{P}$ are also provided in the same table.

\begin{figure}[h]
    \centering
    \includegraphics[width=0.8\textwidth]{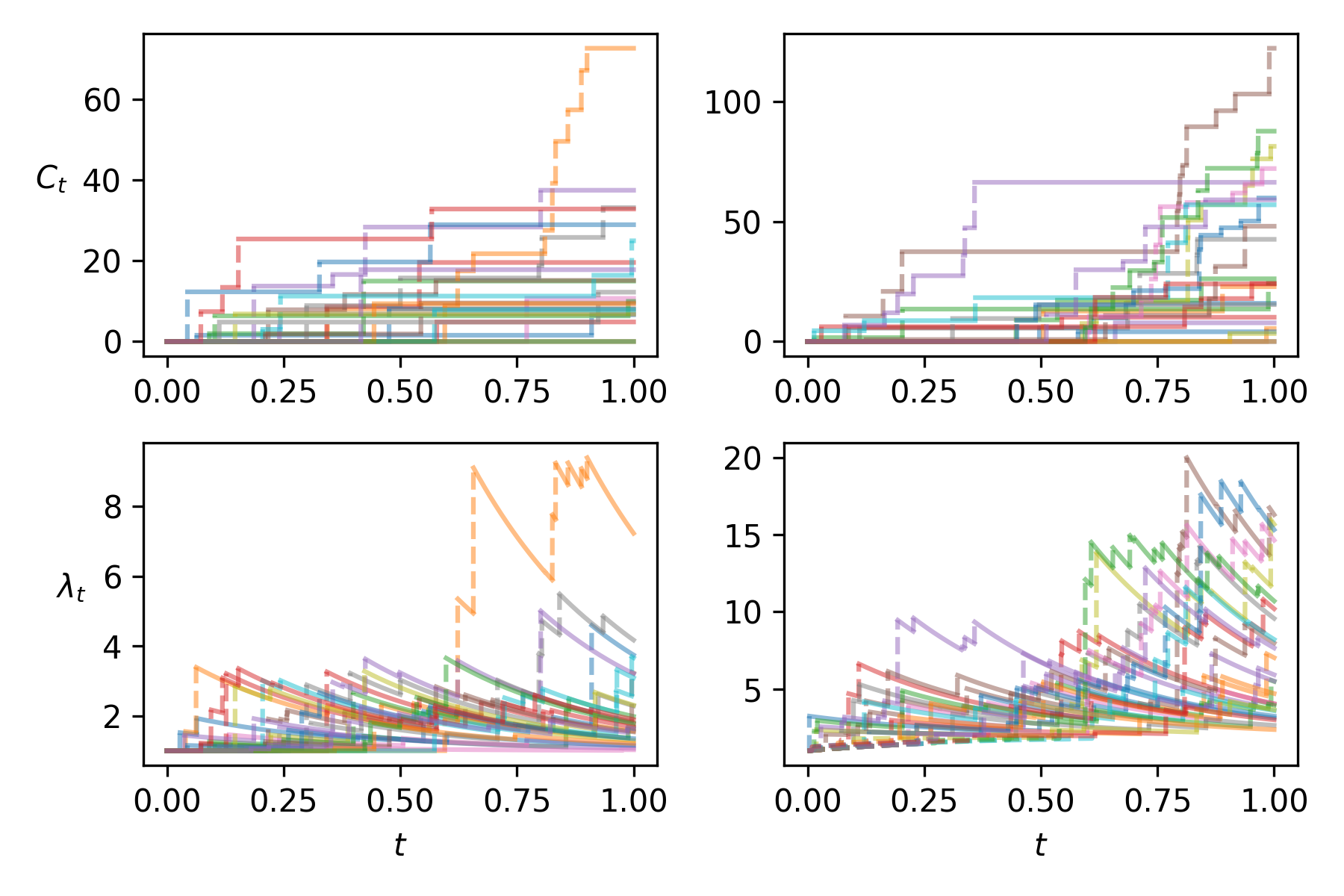}
    \caption{A collection of 25 sample paths of $C_t$ and the corresponding $\lambda_t$ for the compound dynamic contagion process under the original measure $\mathbb{P}$ (left column) and the tilted measure $\tilde{\mathbb{P}}$ (right column) given the constants outlined above. It indicates that paths under $\tilde{\mathbb{P}}$ exhibit higher intensities and larger claims, reflecting risk loading and increased emphasis on extreme losses.}
    \label{Figure 1}
\end{figure}

\begin{table}[H] 
\centering
    \begin{tabular}{rrr}
    \toprule
    $L$ & $\widehat{\mathbb{E}}[(C_t - L)^+]$ & $\widehat{\tilde{\mathbb{E}}}[(C_t - L)^+]$ \\
    \midrule
    0.00 & 13.867646 & 37.635418 \\
    25.00 & 2.556228 & 18.752166 \\
    37.64 & 1.011775 & 12.792098 \\
    50.00 & 0.406231 & 8.738347 \\
    75.00 & 0.062353 & 4.023879 \\
    100.00 & 0.010058 & 1.852134 \\
    \bottomrule
    \end{tabular}
    \caption{The stop-loss CMC estimates  $\widehat{\mathbb{E}}[(C_t - L)^+]$ under the original measure and  $\widehat{\tilde{\mathbb{E}}}[(C_t - L)^+]$ under the Esscher measure for different $L$ retention levels.}
    \label{tbl:dcp}
\end{table}

\Cref{tbl:dcp} shows that arbitrage-free catastrophe stop-loss reinsurance (i.e. gross) premiums are significantly higher than the corresponding net premiums (see \cref{Figure 1}). Depending on the choice of parameters, the gross premiums could be either lowered or elevated. It is not the purpose of this paper to determine which premium is the most appropriate; rather, the insurance companies' risk preferences dictate which equivalent martingale probability measure should be adopted. 

The simplest way to obtain a gross premium is through the use of the parameter~$\theta$, a well-known concept in insurance where a security loading is added to the expected value under the original probability measure~$\mathbb{P}$. An appealing feature of the Esscher transform is that it guarantees the existence of at least one equivalent martingale probability measure in incomplete markets.

\subsection{Sensitivity analysis for the Esscher parameters}

To examine the sensitivity analyses for the Esscher parameters, we focus on the mean, $\widehat{\tilde{\mathbb{E}}}[C_t]$ and the stop-loss retention $L = 25$ for the above combination of parameters. Hence we have:
\begin{equation*}
\widehat{\tilde{\mathbb{E}}}[C_t] = 38.152252  \quad \text{ and } \quad  \widehat{\tilde{\mathbb{E}}}[(C_t - 25)^+] = 19.153988 \,,
\end{equation*}
where $\widehat{\tilde{\mathbb{E}}}[C_t]$ itself can be used for arbitrage-free catastrophe insurance premiums.

The following tables show the way in which these values change when the Esscher parameters are adjusted individually, with all other parameters held constant (except $\phi$ which is updated so that it always satisfies $\phi = - (\theta \hat{j}(\nu) - 1)$).

\subsubsection{\texorpdfstring{Changing $\theta$}{Changing theta}}

\begin{table}[H]
    \centering
    \begin{tabular}{rrrr}
    \toprule
    $\theta$ & $\tilde{\mathbb{E}}[C_t]$ & $\widehat{\tilde{\mathbb{E}}}[C_t]$ & $\widehat{\tilde{\mathbb{E}}}[(C_t - 25)^+]$ \\
    \midrule
    1.00 & 28.137195 & 28.12 ± 0.16 & 11.27 ± 0.13 \\
    1.25 & 37.757126 & 37.64 ± 0.21 & 18.75 ± 0.18 \\
    1.50 & 49.413007 & 49.25 ± 0.27 & 28.80 ± 0.25 \\
    1.75 & 63.671363 & 63.58 ± 0.35 & 41.95 ± 0.33 \\
    \bottomrule
    \end{tabular}
    \caption{The expected loss $\tilde{\mathbb{E}}[C_t]$, the CMC loss $\widehat{\tilde{\mathbb{E}}}[C_t]$ and stop-loss $\widehat{\tilde{\mathbb{E}}}[(C_t - 25)^+]$ estimates (with 95\% confidence intervals), for different values of the tilting parameter $\theta$.}
    \label{tbl:change-theta}
\end{table}

\subsubsection{\texorpdfstring{Changing $\psi$}{Changing psi}}

\begin{table}[H]
    \centering
    \begin{tabular}{rrrr}
    \toprule
    $\psi$ & $\tilde{\mathbb{E}}[C_t]$  & $\widehat{\tilde{\mathbb{E}}}[C_t]$ & $\widehat{\tilde{\mathbb{E}}}[(C_t - 25)^+]$ \\
    \midrule
    1.00 & 34.775140 & 34.84 ± 0.21 & 16.69 ± 0.17 \\
    1.25 & 37.757126 & 37.64 ± 0.21 & 18.75 ± 0.18 \\
    1.50 & 40.737259 & 40.67 ± 0.22 & 21.13 ± 0.19 \\
    1.75 & 43.718045 & 43.83 ± 0.23 & 23.66 ± 0.21 \\
    \bottomrule
    \end{tabular}
    \caption{The expected loss $\tilde{\mathbb{E}}[C_t]$, the CMC loss $\widehat{\tilde{\mathbb{E}}}[C_t]$ and stop-loss $\widehat{\tilde{\mathbb{E}}}[(C_t - 25)^+]$ estimates (with 95\% confidence intervals), for different values of the tilting parameter $\psi$.}
    \label{tbl:change-psi}
\end{table}

\subsubsection{\texorpdfstring{Changing $\nu$}{Changing nu}}

\begin{table}[H]
    \centering
    \begin{tabular}{rrrr}
    \toprule
    $\nu$ & $\tilde{\mathbb{E}}[C_t]$ & $\widehat{\tilde{\mathbb{E}}}[C_t]$ & $\widehat{\tilde{\mathbb{E}}}[(C_t - 25)^+]$ \\
    \midrule
    -0.01 & 22.322487 & 22.39 ± 0.13 & 7.13 ± 0.09 \\
    -0.05 & 37.757126 & 37.64 ± 0.21 & 18.75 ± 0.18 \\
    -0.08 & 62.101734 & 61.90 ± 0.34 & 40.47 ± 0.32 \\
    -0.10 & 93.448189 & 93.15 ± 0.52 & 70.32 ± 0.51 \\
    \bottomrule
    \end{tabular}
    \caption{The expected loss $\tilde{\mathbb{E}}[C_t]$, the CMC loss $\widehat{\tilde{\mathbb{E}}}[C_t]$ and stop-loss $\widehat{\tilde{\mathbb{E}}}[(C_t - 25)^+]$ estimates (with 95\% confidence intervals), for different values of the tilting parameter $\nu$.}
    \label{tbl:change-nu}
\end{table}

Comparing the numerical values in \cref{tbl:change-theta} with those in \cref{tbl:change-psi}, the increases in gross insurance and reinsurance premium estimates due to changes in $\theta$ are greater than those resulting from changes in $\psi$. This is because $\hat{j}(\nu)$ is involved in the self-exciting intensity function $\lambda_{t}$. \Cref{tbl:change-nu} illustrates the increase in gross insurance and reinsurance premium estimates resulting from changes in $\nu$, which is due to both a higher self-exciting intensity and larger claim/loss sizes.

\subsection{Sensitivity analysis for the intensity parameters}

We next examine sensitivity with respect to the intensity parameters governing the
dynamics of the claim arrival process. Using the same setup (i.e.\ the same baseline
parameter values) as in the previous subsection, we focus on the mean
$\widehat{\tilde{\mathbb{E}}}[C_t]$ and the stop-loss premium with retention $L = 25$.

The following tables illustrate how these quantities change when key intensity
parameters are varied individually, with all other parameters held constant
(except $\phi$, which is updated to satisfy $\phi = -(\theta \hat{j}(\nu) - 1)$).

\subsubsection{\texorpdfstring{Changing $\delta$}{Changing delta}}

\begin{table}[H]
    \centering
    \begin{tabular}{rrrr}
    \toprule
    $\delta$ & $\tilde{\mathbb{E}}[C_t]$ & $\widehat{\tilde{\mathbb{E}}}[C_t]$ & $\widehat{\tilde{\mathbb{E}}}[(C_t - 25)^+]$ \\
    \midrule
    2.00 & 43.634528 & 43.60 ± 0.26 & 24.23 ± 0.23 \\
    3.00 & 37.757126 & 37.64 ± 0.21 & 18.75 ± 0.18 \\
    5.00 & 32.620608 & 32.69 ± 0.18 & 14.26 ± 0.14 \\
    7.00 & 58.289065 & 58.21 ± 0.37 & 37.36 ± 0.35 \\
    \bottomrule
    \end{tabular}
    \caption{The expected loss $\tilde{\mathbb{E}}[C_t]$, the CMC loss $\widehat{\tilde{\mathbb{E}}}[C_t]$ and stop-loss $\widehat{\tilde{\mathbb{E}}}[(C_t - 25)^+]$ estimates (with 95\% confidence intervals), for different values of the tilting parameter $\delta$.}
    \label{tbl:change-delta}
\end{table}

\subsubsection{\texorpdfstring{Changing $\alpha$}{Changing alpha}}

\begin{table}[H]
    \centering
    \begin{tabular}{rrrr}
    \toprule
    $\alpha$ & $\tilde{\mathbb{E}}[C_t]$  & $\widehat{\tilde{\mathbb{E}}}[C_t]$ & $\widehat{\tilde{\mathbb{E}}}[(C_t - 25)^+]$ \\
    \midrule
    1.00 & 53.134574 & 53.19 ± 0.27 & 32.01 ± 0.25 \\
    2.00 & 37.757126 & 37.64 ± 0.21 & 18.75 ± 0.18 \\
    3.00 & 32.736742 & 32.74 ± 0.20 & 15.03 ± 0.16 \\
    4.00 & 30.246361 & 30.30 ± 0.19 & 13.28 ± 0.15 \\
    \bottomrule
    \end{tabular}
    \caption{The expected loss $\tilde{\mathbb{E}}[C_t]$, the CMC loss $\widehat{\tilde{\mathbb{E}}}[C_t]$ and stop-loss $\widehat{\tilde{\mathbb{E}}}[(C_t - 25)^+]$ estimates (with 95\% confidence intervals), for different values of the tilting parameter $\alpha$.}
    \label{tbl:change-alpha}
\end{table}

\subsubsection{\texorpdfstring{Changing $\beta$}{Changing beta}}

\begin{table}[H]
    \centering
    \begin{tabular}{rrrr}
    \toprule
    $\beta$ & $\tilde{\mathbb{E}}[C_t]$ & $\widehat{\tilde{\mathbb{E}}}[C_t]$ & $\widehat{\tilde{\mathbb{E}}}[(C_t - 25)^+]$ \\
    \midrule
    1.00 & 37.757126 & 37.64 ± 0.21 & 18.75 ± 0.18 \\
    1.50 & 31.529562 & 31.50 ± 0.16 & 13.08 ± 0.13 \\
    2.00 & 29.013963 & 29.04 ± 0.15 & 10.90 ± 0.11 \\
    2.50 & 27.659677 & 27.73 ± 0.14 & 9.77 ± 0.10 \\
    \bottomrule
    \end{tabular}
    \caption{The expected loss $\tilde{\mathbb{E}}[C_t]$, the CMC loss $\widehat{\tilde{\mathbb{E}}}[C_t]$ and stop-loss $\widehat{\tilde{\mathbb{E}}}[(C_t - 25)^+]$ estimates (with 95\% confidence intervals), for different values of the tilting parameter $\beta$.}
    \label{tbl:change-beta}
\end{table}

\Cref{tbl:change-delta} shows the complicated relationship that $\delta$ has on the expected losses.
Changing $\delta$ changes the solution $B(t)$, and that changes nearly all of the Esscher-transformed parameters (except for claim sizes), making the impacts non-monotonic.
\cref{tbl:change-alpha,tbl:change-beta} present the reduction in them as the external-exciting and self-exciting jump size parameters \(\alpha\) and \(\beta\) increase. An increase in these jump size parameters results in a decrease in the exponential average.

\subsection{Comparison with Generalised Hawkes and Shot-noise Cox processes}
If there are no externally excited jumps 
in \cref{eq:infinitesimal_generator} (i.e., $\rho =0$) $C_{t}$
becomes generalised compound Hawkes process. If $a=0$ and there are no
self-excited jumps in \cref{eq:infinitesimal_generator}, $C_{t}$ 
becomes compound Cox process with shot-noise Poisson intensity.

The following two tables show the corresponding stop-loss CMC estimates in \cref{tbl:dcp} (DCP case) for the generalised Hawkes case and the Cox case, respectively. 

\subsubsection{Generalised Hawkes process}

If we simulate using the same parameters as above, except setting $\rho = 0$, then the underlying point process becomes a generalised Hawkes process.

\begin{table}[H] 
    \centering
    \begin{tabular}{rrr}
    \toprule
    $L$ & $\widehat{\mathbb{E}}[(C_t - L)^+]$ & $\widehat{\tilde{\mathbb{E}}}[(C_t - L)^+]$ \\
    \midrule
    0.00 & 9.643789 & 22.869498 \\
    25.00 & 1.251784 & 8.592431 \\
    37.64 & 0.441185 & 5.278698 \\
    50.00 & 0.160856 & 3.318908 \\
    75.00 & 0.019958 & 1.351518 \\
    100.00 & 0.002533 & 0.573990 \\
    \bottomrule
    \end{tabular}
    \caption{The stop-loss CMC estimates  $\widehat{\mathbb{E}}[(C_t - L)^+]$ under the original measure and  $\widehat{\tilde{\mathbb{E}}}[(C_t - L)^+]$ under the Esscher measure for different $L$ retention levels. The point process driving $C_t$ is a generalised Hawkes process.}
    \label{tbl:hawkes}
\end{table}

\subsubsection{Shot-noise Cox process}

If we simulate using the same parameters as above, except setting $a = 0$ and $\beta = \infty$ (i.e., $Y_j \equiv 0$ for all $j$), then the underlying point process becomes a shot-noise Cox process.

\begin{table}[H] 
    \centering
    \begin{tabular}{rrr}
    \toprule
    $L$ & $\widehat{\mathbb{E}}[(C_t - L)^+]$ & $\widehat{\tilde{\mathbb{E}}}[(C_t - L)^+]$ \\
    \midrule
    0.00 & 5.804331 & 12.216566 \\
    25.00 & 0.231064 & 1.840665 \\
    37.64 & 0.035712 & 0.608449 \\
    50.00 & 0.004792 & 0.192236 \\
    75.00 & 0.000030 & 0.017000 \\
    100.00 & 0.000000 & 0.001662 \\
    \bottomrule
    \end{tabular}
    \caption{The stop-loss CMC estimates  $\widehat{\mathbb{E}}[(C_t - L)^+]$ under the original measure and  $\widehat{\tilde{\mathbb{E}}}[(C_t - L)^+]$ under the Esscher measure for different $L$ retention levels. The point process driving $C_t$ is a shot-noise Cox process.}
    \label{tbl:cox}
\end{table}

Compared to \cref{tbl:hawkes,tbl:cox}, \cref{tbl:dcp} shows that gross insurance and reinsurance premium estimates significantly increase when changing $C_{t}$ from the Cox process with shot-noise Poisson intensity to a generalised Hawkes process due to
self-excited jumps, and to the dynamic contagion process due to both externally excited jumps and self-excited jumps.

\section{Conclusion}\label{Sec_Conclusion}

The insurance sector is an important component of our economy, with trillions of dollars under management. We rely on well-functioning insurers. However, the current pricing mechanism for catastrophic risk insurance has put pressure on insurers and reinsurers when faced with
more frequent and larger natural and man-made disasters. For example, in the aftermath of frequent natural disasters such as the 2022 South-East Queensland and New South Wales floods, many insurers only offer policies that include flood cover with approximately 30\% increased premiums. Before the Southern California wildfires of January 2025, State Farm did not renew insurance policies in high-risk California postcodes due to increasing costs.

Unless (re)insurers have an appropriate framework and prudent methodology which adequately assess catastrophe risks, ordinary individuals and business owners will suffer as the industry will not be able to promptly and fully meet the many claims of catastrophe victims. To address this challenge, in this paper, we presented how to obtain arbitrage-free catastrophe stop-loss reinsurance premiums when dealing with the ongoing challenge and complexity of emerging risk dynamics. To predict claim/loss arrivals from conventional and emerging catastrophes, we used the compound dynamic contagion process for the catastrophic
component of the liability and the Esscher transform to determine arbitrage-free premiums. Sensitivity analyses were also performed by changing the retention level, the Esscher parameters and the intensity parameters. Given the broad applicability of the CDCP, it is expected that what we have obtained in this paper will provide practitioners with feasible approaches to quantify their catastrophic liabilities in light of
the growing challenges posed by new risks arising from climate change, cyberattacks, and pandemics. We also envisage that the dynamics used in
this paper extend naturally for pricing a variety of catastrophe insurance derivatives. For this, we leave them as objects of further research.

\paragraph{Competing Interests}
The authors declare none.

\paragraph{Data Availability Statement}
The Python code used to generate the numerical results is available at: \url{https://github.com/Pat-Laub/ReinsuranceEsscherDynamicContagion}.

\paragraph{Funding Statement}
This work received no specific grant from any funding agency, commercial or not-for-profit sectors.

\bibliographystyle{apa}

\newpage
\appendix
\renewcommand{\thesection}{\Alph{section}}

\section{Equivalent martingale measure with non-zero interest rate} \label{EMM Non-zero}

\begin{definition}[Equivalent martingale measure with discounting]
\label{Emm_discounted}
Let $r \ge 0$ be a constant risk-free interest rate. 
A probability measure $\tilde{\mathbb{P}}$ is said to be an equivalent martingale probability measure 
(i.e., $\tilde{\mathbb{P}}$ is equivalent to $\mathbb{P}$ on ${\cal F}_T$) if it satisfies the following properties:
\begin{enumerate}[label=(\roman*)]
    \item $\tilde{\mathbb{P}}(A)=0$ if and only if $\mathbb{P}(A)=0$ for any $A \in {\cal F}_T$;
    \item The Radon--Nikodym derivative $\frac{\dd \tilde{\mathbb{P}}}{\dd \mathbb{P}} \in L^{2}(\Omega, {\cal F}_{T}, \mathbb{P})$;
    \item The discounted surplus process $\{ \ee^{-rt} R_t \}_{t \in {\cal T}}$ is an $(\mathbb{F}, \tilde{\mathbb{P}})$-martingale. 
    In particular,
    \[
    \tilde{\mathbb{E}}\!\left[ \ee^{-rt} R_t \mid {\cal F}_s \right]
    =
    \ee^{-rs} R_s,
    \qquad 0 \le s \le t \le T,
    \quad \tilde{\mathbb{P}}\text{-a.s.}
    \]
\end{enumerate}
\end{definition}

\begin{remark}
Assuming $R_0=0$, \cref{Emm_discounted} (iii) implies
\[
\tilde{\mathbb{E}}\!\left( \ee^{-rt} R_t \mid \mathcal{F}_0 \right)
=
\tilde{\mathbb{E}}\!\left( \ee^{-rt} (P_t - C_t) \mid \mathcal{F}_0 \right)
=0.
\]
Consequently,
\[
\tilde{\mathbb{E}}\!\left( \ee^{-rt} C_t \mid \mathcal{F}_0 \right)
=
\tilde{\mathbb{E}}\!\left( \ee^{-rt} P_t \mid \mathcal{F}_0 \right),
\]
which characterises the arbitrage-free insurance premium under a non-zero constant interest rate.
\end{remark}

\section{Simulation algorithms} \label{simulation}

\cref{simulate-cdcp} shows our method for simulating the time-inhomogeneous compound dynamic contagion process, which extends the work of \cite{ogata1981lewis}.
It utilises \cref{simulate-poisson} which is the well-known thinning algorithm for Poisson processes (see \citealp{lewis1979simulation}).

\begin{algorithm}[H]
\caption{Simulate a time-inhomogeneous Poisson process}
\begin{algorithmic}[1]
\Procedure{SampleInhomogeneousPoissonProcess}{$T,\; \rho(\cdot)$}
    \State Compute \(\rho_{\max} \gets \max_{t \in [0,T]} \rho(t)\)
    \State Sample \(N \sim \mathrm{Poisson}(\rho_{\max} \cdot T)\)
    \State Draw \(N\) candidate times \(\{t^*_1,\ldots,t^*_N\}\) uniformly in \([0,T]\) and sort them
    \State Initialise \(\mathcal{T} \gets \emptyset\)
    \For{each candidate time \(t^*\) in \(\{t^*_1,\ldots,t^*_N\}\)}
         \State Sample \(U \sim \mathrm{Uniform}(0,1)\)
         \If{\(U < \rho(t^*)/\rho_{\max}\)}
              \State Append \(t^*\) to \(\mathcal{T}\)
         \EndIf
    \EndFor
    \State \Return \(\mathcal{T}\)
\EndProcedure
\end{algorithmic}
\label{simulate-poisson}
\end{algorithm}

\noindent
In the next algorithm we split the terms in \cref{idcp-definition} into
\begin{equation} \label{idcp-definition-split}
\lambda_{t} = \underbrace{\lambda_0 \ee^{-\delta t} + \delta \int_{0}^{t} a(s) \ee^{-\delta (t-s)} \, \dd s + \sum_{i \ge
1} X_{i} \ee^{-\delta \left( t - T_{1, i} \right)}\mathbb{I}_{\{ T_{1,i} \le t \}} }_{=: \lambda^{(c)}_t , \text{ the ``Cox'' part}} + \underbrace{\sum_{j \ge 1} Y_{j} \ee^{-\delta \left(t - T_{2,j} \right)}
\mathbb{I}_{ \{ T_{2, j} \le t \}}}_{=: \lambda^{(h)}_t , \text{ the ``Hawkes'' part}} \,.
\end{equation}

\begin{algorithm}[H]
\caption{Simulate a time-inhomogeneous compound dynamic contagion process}
\begin{algorithmic}[1]
\Procedure{SimulateCDCP}{$\mathcal{T} = [0, T],\; \lambda_0,\; a(t),\; \rho(t),\; \delta,\; G(\cdot;t),\; H(\cdot;t),\; J(\cdot),\; \Delta t_{\text{max}}$}
    \State
    \State \textbf{// Step 1: Sample External (Cox) Arrival Times and Jump Sizes}
    \State Let \(\{t_{1,1}, t_{1,2}, \ldots, t_{1,n}\} \gets \textsc{SampleInhomogeneousPoissonProcess}(T,\; \rho(\cdot))\)
    \ForAll{\(i = 1,2,\ldots, n\)}
         \State Sample \(X_i \sim H(\,\cdot\,; t_{1,i})\)
    \EndFor
    \State
    \State \textbf{// Step 2: Pre-compute the Cox Component of the Intensity}
    \State Compute
    $
    \lambda^{(c)}_t \gets \lambda_0 \ee^{-\delta t} + \delta \int_{0}^{t} a(s) \ee^{-\delta (t-s)} \, \dd s + \sum_{i \ge
1} X_{i} \ee^{-\delta \left( t - T_{1, i} \right)}\mathbb{I}_{\{ T_{1,i} \le t \}}
    $ for all \(t \in \mathcal{T}\)
    \State
    \State \textbf{// Step 3: Simulate Self-Excited (Hawkes) Arrivals via Ogata's Thinning}
    \State Initialise \(\lambda^{(h)}_t \gets 0\) for all \(t \in \mathcal{T}\), \(t \gets 0\) and \(j \gets 0\)
    \Comment {Note that $\lambda_t = \lambda^{(c)}_t + \lambda^{(h)}_t$}
    \State
    \State Set \(t \gets 0\) and \(j \gets 0\)
    \While{\(t \le T\)}
         \State \textbf{// Generate a potential arrival time}
         \State Compute
         $
         \lambda_{\max} \gets \max_{t' \in [t, (t+\Delta t_{\text{max}}) \wedge T]} \  \lambda^{(c)}_{t'} + \lambda^{(h)}_{t'}  
         $
         \State Sample \(\Delta t \sim \mathrm{Exponential}(\lambda_{\max})\)
         \State Update \(t \gets t + (\Delta t \wedge \Delta t_{\text{max}})\)
         \State
         \State \textbf{// Stop simulating if we reach the end of the simulation horizon}
         \If{\(t > T\)} 
              \State \textbf{break}
         \EndIf
         \State
         \State \textbf{// Discard proposals that are too far in the future}
         \If{\(\Delta t > \Delta t_{\text{max}}\)}
              \State \textbf{continue}
         \EndIf
         \State
         \State \textbf{// Keep the proposal with probability $\lambda_t / \lambda_{\text{max}}$}
         \State Sample \(U \sim \mathrm{Uniform}(0,1)\)
         \If{\(U < (\lambda^{(c)}_t + \lambda^{(h)}_t) /\lambda_{\max}\)}
              \State \textbf{// Update various counters \& intensities and sample a jump size}
              \State \(j \gets j + 1\) and \(t_{2,j} \gets t \)
              \State Sample \(Y_j \sim G(\,\cdot\,; t_{2,j} )\) \Comment{Next Hawkes intensity jump}
              \State Update \(\lambda^{(h)}_{t'} \gets \lambda^{(h)}_{t'} + Y_j\, \ee^{-\delta (t' - t_{2,j})}\) for all $t' \in [t_{2,j}, T]$
            
              \State Sample \(\Xi_j \sim J(\,\cdot\,)\) \Comment{Next jump in the CDCP}
         \EndIf
    \EndWhile
    \State
    \State \Return \(C_t = \sum_{j \ge 1} \Xi_j\, \mathbb{I}_{\{t_{2,j} \le t\}}\) for all \(t \in \mathcal{T}\)
\EndProcedure
\end{algorithmic}
\label{simulate-cdcp}
\end{algorithm}

\end{document}